\documentclass[11pt]{article}
\usepackage[utf8]{inputenc}

\usepackage{amsthm,amsmath,bm,bbm}
\usepackage{amssymb,mathtools}
\usepackage{nccmath}
\usepackage[paper=letterpaper,margin=1in]{geometry}
\usepackage{xcolor}
\usepackage[colorlinks=true, allcolors=blue]{hyperref}
\usepackage{wrapfig}
\usepackage{multirow}
\usepackage{caption}
\usepackage{booktabs}

\usepackage{enumerate}
\usepackage{enumitem}

\usepackage[compact]{titlesec}

\usepackage{natbib}
\setcitestyle{authoryear, open={(},close={)}}
\usepackage{minitoc}
\newtheorem{lemma}{Lemma}

\usepackage{setspace}
\usepackage{tikz}
\usetikzlibrary{shapes,shadows,arrows,positioning}
\def\independent{\perp\!\!\!\perp}
\def\var{\text{var}}

\def\E{\text{E}}
\def\P{\text{P}}

\title{\Large Identification of complier and noncomplier average causal effects\\in the presence of \textit{latent} missing-at-random (LMAR) outcomes:\\a unifying view and choices of assumptions}

\author{Trang Quynh Nguyen, Michelle C. Carlson, and Elizabeth A. Stuart\\[.5em]\normalsize Johns Hopkins Bloomberg School of Public Health}
\date{}

\begin{document}

\maketitle

\begin{abstract}
    \noindent The study of treatment effects is often complicated by noncompliance and missing data.
    In the one-sided noncompliance setting where of interest are the complier and noncomplier average causal effects (CACE and NACE), we address outcome missingness of the \textit{latent missing at random} type (LMAR, also known as \textit{latent ignorability}). That is, conditional on covariates and treatment assigned, the missingness may depend on compliance type. 
    Within the instrumental variable (IV) approach to noncompliance, methods have been proposed for handling LMAR outcome that additionally invoke an exclusion restriction type assumption on missingness, but no solution has been proposed for when a non-IV approach is used.
    This paper focuses on effect identification in the presence of LMAR outcome, with a view to flexibly accommodate different principal identification approaches.
    We show that under treatment assignment ignorability and LMAR only, effect nonidentifiability boils down to a set of two connected mixture equations involving unidentified stratum-specific response probabilities and outcome means. 
    This clarifies that (except for a special case) effect identification generally requires two additional assumptions: a \textit{specific missingness mechanism} assumption and a \textit{principal identification} assumption.
    This provides a template for identifying effects based on separate choices of these assumptions.
    We consider a range of specific missingness assumptions, including those that have appeared in the literature and some new ones. Incidentally, we find an issue in the existing assumptions, and propose a modification of the assumptions to avoid the issue.
    Results under different assumptions are illustrated using data from the Baltimore Experience Corps Trial.

    ~

\end{abstract}

% \doparttoc % Tell to minitoc to generate a toc for the parts
% \faketableofcontents % Run a fake tableofcontents command for the partocs

% \part{} % Start the document part
% \parttoc % Insert the document TOC

% \tableofcontents

% \newpage
% \clearpage
\pagenumbering{arabic}
% \doublespacing

\section{Introduction}\label{sec:intro}

The study of treatment effects is often complicated by noncompliance and missing data. 
This paper considers the one-sided noncompliance setting, where noncompliance occurs in the treatment but not control condition.
Here in addition to (or instead of) the average treatment effect (ATE), 
we are interested in the causal effects of treatment assignment on two types of study participants, aka \textit{principal strata} \citep{frangakis2002PrincipalStratificationCausal}: those who \textit{would} and those who \textit{would not} comply to the active treatment \textit{if they were offered it}. (What to comply means is specific to the study, e.g., participating in a one-off intervention, or 80\% adherence to a medication.) For brevity, we refer to these two types as \textit{compliers} and \textit{noncompliers}, although it might help to think of them as \textit{would-be} compliers and noncompliers.
While everyone belongs in one type, the type is not observed in the control condition, posing a challenge for identifying the complier and noncomplier average causal effects (CACE and NACE), also known as \textit{principal causal effects} (PCEs). These effects require special identification assumptions beyond standard causal inference assumptions.

Outcome missingness is ubiquitous \citep{wood2004AreMissingOutcome} and further complicates effect identification. 
This poses the double challenge of (i) identifying the PCEs (had there been no missing data) and (ii) recovering identification under outcome missingness. This paper focuses on (ii) and specifically addresses a type of outcome missingness termed \textit{latent missing at random} (LMAR, explained shortly).
To set the stage, we briefly mention major approaches to PCE identification before reviewing assumptions that have been used to handle outcome missingness within those approaches.

% \paragraph{Approaches to identifying principal causal effects.}
There are several approaches that supplement standard causal inference assumptions with assumptions that specifically handle the fact that compliance type is only partially observable -- to achieve identification of PCEs in the absence of missing data. We consider two major approaches. One treats the assigned treatment as an instrumental variable (IV) that supposedly satisfies the \textit{exclusion restriction} (ER) assumption: it does not affect outcome other than through affecting treatment received \citep{Angrist1995}. This means the NACE is zero, and the CACE explains the full effect of treatment assignment. 
The second approach does not restrict the NACE to zero but instead leverages covariates to invoke the \textit{principal ignorability} (PI) assumption: conditional on baseline covariates, principal stratum is ignorable for outcome under control, or (put another way) compliers and noncompliers share the same outcome distribution under control \citep{stuart2015AssessingSensitivityMethods,feller2017PrincipalScoreMethods,ding2017PrincipalStratificationAnalysis}. 
As ER and PI are untestable, sensitivity analyses are called for.
In this paper we will also include several sensitivity assumptions deviating from PI that point identify the CACE and NACE \citep{ding2017PrincipalStratificationAnalysis,nguyen2023SensitivityAnalysisPrincipal}. We refer to these assumptions (ER, PI, PI replacement) as \textit{principal identification assumptions}.

\subsection{Missingness assumptions}

With outcome missingness, the assumptions above no longer suffice, and some assumptions about the missingness are needed to recover effect identification. We consider two layers of missingness assumptions.

\subsubsection{Latent missing-at-random as a general assumption}

\cite{frangakis1999AddressingComplicationsIntentiontotreat} proposed the \textit{latent ignorability} assumption: the missingness is independent of the outcome value given the assigned treatment and compliance type. (``Latent'' refers to the fact that compliance type is not observed for those assigned to control.) As a starting point in reasoning about missingness, this assumption is appealing in that it allows the missingness to depend on compliance type in \textit{both} treatment conditions, whereas a missing at random (MAR) assumption does not allow this in the control condition. 
\cite{mealli2004AnalyzingRandomizedTrial} stated this assumption in a form (which we adopt in this paper) that explicitly allows the missingness to depend on baseline covariates in addition to treatment assigned and compliance type. While \textit{latent ignorability} is a fine label, we opt to call the assumption \textit{latent missing-at-random} (LMAR) \citep[following][]{peng2004ExtendedGeneralLocation,beesley2019SequentialImputationModels}. This label signals that the assumption is a relaxation of a MAR assumption. It also says that the assumption is about missingness, clearly differentiating it from other assumptions that are also called ``ignorability'' such as treatment assignment ignorability and principal ignorability.

Note that LMAR means that the outcome is missing at random (MAR) in the treatment condition but missing not at random (MNAR) in the control condition. Hence LMAR is generally not sufficient to recover effect identification.

\subsubsection{Specific missingness mechanisms}

For an analysis in the IV/ER approach, to recover CACE identification, \cite{frangakis1999AddressingComplicationsIntentiontotreat} supplemented LMAR with a second ER assumption on \textit{response}, calling the combination of the two ER assumptions ``compound ER.'' \textit{Response} here means the outcome is observed, i.e., the opposite of missing. % 
\citeauthor{frangakis1999AddressingComplicationsIntentiontotreat} stated the assumption in a deterministic form: if a noncomplier's outcome is to be observed (missing), it would be observed (missing) under either treatment assignment. 
\cite{mealli2004AnalyzingRandomizedTrial} instead used a stochastic version of this assumption that only requires that, for noncompliers, missingness is independent of treatment assignment given covariates. \citeauthor{mealli2004AnalyzingRandomizedTrial} also propose that in some settings one might prefer the alternative assumption that this conditional independence holds among compliers rather than noncompliers; they call this \textit{response ER for compliers} to differentiate from the original assumption, although it is not really an ER assumption. \cite{jo2010HandlingMissingData} called the original assumption \textit{response ER} (rER) and \citeauthor{mealli2004AnalyzingRandomizedTrial}'s alternative assumption \textit{stable complier response} (SCR). We adopt \citeauthor{jo2010HandlingMissingData}'s terminology, and will also refer to rER as \textit{stable noncomplier response} (SNR).

To our knowledge, the missing outcome problem has not been explicated within other principal identification approaches. Methods papers on PI and on sensitivity analysis for PI violation have so far put this problem aside. Some include a small note that MAR outcomes are simple to handle, e.g., by inverse probability of response weighting \citep{feller2017PrincipalScoreMethods,ding2017PrincipalStratificationAnalysis,nguyen2023SensitivityAnalysisPrincipal}, while non-ignorable missingness (e.g., LMAR) is a complex topic for future research \citep{ding2017PrincipalStratificationAnalysis,nguyen2023SensitivityAnalysisPrincipal}.

\begin{table}
\caption{Acronyms and of shorthand math notations}\label{tab:acronyms}
\centering
\resizebox{.8\linewidth}{!}{%
\begin{tabular}{llcll}
    \multicolumn{2}{l}{Acronyms}
    &~~~~~~~~~&
    \multicolumn{2}{l}{Mathematical shorthands}
    \\\hline
    ATE & average treatment effect
    && $\Delta$ & $\E[Y_1-Y_0]$
    \\
    CACE & complier average causal effect
    && $\Delta_c$ & $\E[Y_1-Y_0\mid C=c]$
    \\
    NACE & noncomplier average causal effect
    && $\delta_c(X)$ & $\E[Y_1-Y_0\mid X,C=c]$
    \\
    PCEs & principal causal effects
    \\
    &&& $\pi_c(X)$ & $\P(C=c\mid X,Z=1)$
    \\
    IV & instrumental variable
    && $\pi_{0c}^R(X)$ & $\P(C=c\mid X,Z=0,R=1)$
    \\
    ER & exclusion restriction
    \\
    PI & principal ignorability
    && $\mu_{zc}(X)$ & $\E[Y\mid X,Z=x,C=c]$
    \\
    PIsens & sensitivity to PI violation
    && $\kappa_0(X)$ & $\E[Y\mid X,Z=0]$
    \\
    GOR & generalized odds ratio
    && $\kappa_0^R(X)$ & $\E[Y\mid X,Z=0,R=1]$
    \\
    MR & mean ratio
    \\
    SMD & standardized mean difference
    && $\sigma_{0c}^2(X)$ & $\var(Y\mid X,Z=0,C=c)$
    \\
    &&& $\varsigma_0^{2R}(X)$ & $\var(Y\mid X,Z=0,R=1)$
    \\
    LMAR & latent missing at random
    \\
    MAR & missing at random
    && $\varpi_{zc}(X)$ & $\P(R=1\mid X,Z=z,C=c)$
    \\
    MNAR & missing not at random
    && $\lambda_0(X)$ & $\P(R=1\mid X,Z=0)$
    \\
    \\
    rER & response ER
    \\
    (near-)SNR & (near) stable noncomplier response
    \\
    (near-)SCR & (near) stable complier response
    \\
    rPI & response PI
    \\
    rPO & proportional response odds
    \\
    \hline
\end{tabular}%
}
\end{table}

\subsection{Our contributions}

To address the gaps of the current method options, our paper provides a more general template for handling the identification challenge due to LMAR outcome. This template flexibly accommodates different principal identification approaches including those commonly used in practice.
The results that have appeared in the literature (restricted to stable response assumptions within the IV/ER principal identification approach) are obtained as specific case solutions, but the template is more general and allows the use of different specific missingness assumptions (not limited to stable response) that identify certain response probabilities.
% This paper aims to shed light on, and suggest solutions to, the identification challenge due to LMAR outcome, with a view to flexibly accommodate different principal identification approaches. 
% (``Identification'' here means \textit{causal identification}, specifically in the form of \textit{point} identification; that is, showing that the estimand, under certain assumptions, is equal to a function of the observed data distribution.)
% In a sense this provides license to attempt to learn it from data. 
% To be precise this specific form of identification is \textit{point} identification; the paper does not cover \textit{partial identification} where there is only enough information to bound the estimand.

To develop the template, we first take a step back and characterize the nonidentifiability of the PCEs under the combination of treatment assignment ignorability and LMAR only -- which we call the \textit{base assumptions} -- before adding specific missingness assumptions and principal identification assumptions. The results are summarized below.
\begin{enumerate}[leftmargin=*]
    \item We show that PCE nonidentifiability under the base assumptions boils down to a set of two mixture equations involving unidentified principal-stratum-specific \textit{response probabilities} and \textit{outcome means} under control given covariates. We call these the \textit{response mixture equation} and the \textit{outcome mixture equation}. They are connected because the unknowns in the first equation also appear in the second equation.
    
    The key insight from this finding, which defines the aforementioned template, is that PCE identification generally requires two additional assumptions: a specific missingness mechanism assumption to solve the response mixture equation and a principal identification assumption to then solve the outcome mixture equation.
    These are separate processes, so assumptions of different flavors can be combined, e.g., the ER principal identification assumption does not need to be paired with an ER-type missingness assumption. The choice of these assumptions should be based on substantive considerations.

    \item To provide immediately useful results, we derive identification formulas when supplementing the base assumptions with different pairings of principal identification and specific missingness assumptions. The principal identification assumptions covered in these results include ER, PI and several assumptions that replace PI in sensitivity analysis.
    The specific missingness assumptions belong in two types. One type compares missingness across treatment conditions; these include the existing stable response assumptions (SNR/rER and SCR) plus a modified version of these assumptions (see point 4 below). The second type, which we propose as an alternative, compares missingness across principal strata; these include a PI-type assumption on response (rPI) and an assumption that lets the missingness under control differ between compliers and noncompliers.

    There are two special results that are simple.
    First, if PI is used as the principal identification assumption, the PCEs are identified without needing to assume a specific missingness mechanism. This is because PI provides a solution to the outcome mixture equation that bypasses the response mixture equation. (Note though that when a PI-based analysis is paired with a sensitivity analysis that deviates from PI, the sensitivity analysis falls out of this special case and requires a specific missingness mechanism.)
    Second, under either PI or rPI (combined with LMAR), the outcome turns out to be MAR.

    \item Importantly, the template we employ can be used to derive identification formulas should other assumptions be chosen, as long as (i) the principal identification assumption would point-identify the effects had the outcome been fully observed, and (ii) the specific missingness assumption solves the response mixture equation.

    % \item Special results: An important special result is that if PI is used as the principal identification assumption, the PCEs are identified without needing to assume a specific missingness mechanism. This is because PI provides a solution to the outcome mixture equation that bypasses the response mixture equation. (Note though that when a PI-based analysis is paired with a sensitivity analysis that deviates from PI, the sensitivity analysis falls out of this special case and requires a specific missingness mechanism.)

    % Also, under PI and rPI (combined with LMAR), the outcome turns out to be MAR.

    \item An incidental finding is that the existing stable response assumptions (SNR/rER and SCR) have a limitation: it is possible for them to imply implausible response probabilities. We propose a modification of these assumptions to avoid this issue.
\end{enumerate}

Note that as this work focuses on LMAR, it does not tackle missingness that depends on the outcome conditional on covariates, treatment assigned and compliance type.

\medskip

The paper proceeds as follows. Section \ref{sec:EC-intro} presents the Experience Corps study concerning a volunteering intervention for the elderly, which will serve as the illustrative example. Section \ref{sec:setting} introduces the setting, notation, estimands, and base assumptions that are made throughout. Section \ref{sec:nonidentifiability} clarifies the non-identifiability of the estimands under the base assumptions alone. Section \ref{sec:solutions} solves the identification problem using different specific missingness assumptions and principal identification assumptions. Section \ref{sec:illustration} uses Experience Corps data to illustrate the identification results. 
% Section \ref{sec:rER} provides additional comment on the stable response assumptions.
Section \ref{sec:conclusions} concludes with a discussion.

\section{An example: the Baltimore Experience Corps Trial}\label{sec:EC-intro}

In the Baltimore Experience Corps Trial \citep{gruenewald2016BaltimoreExperienceCorps}, adults aged 60 or older were randomized either (i) to participate in a program (the \textit{intervention} condition) where they received training and then served as volunteers helping children in elementary schools  or (ii) to be given usual volunteer opportunities available in the community (the \textit{control} condition). Those in the intervention condition varied a great deal in the number of hours volunteered, while volunteering activity in the control condition was not tracked. The study has been analyzed using the principal stratification framework by dichotomizing the number of hours volunteered \citep{gruenewald2016BaltimoreExperienceCorps,parisi2015IncreasesLifestyleActivities}. We follow this convention with a dichotomization at 730 hours of volunteering over 24 months (on average 7 hours per week) to form a \textit{high participation} and a \textit{low participation} group. (These correspond to our generic terms \textit{compliers} and \textit{noncompliers} above.) In the intervention arm, over half of study participants belong in the high participation group and under half in the low participation group. Previous analyses have used ER as the principal identification assumption. For illustrative purposes in this paper we will use other assumptions in addition to ER. 

We will examine one of the study's two main outcomes, \textit{perceived generativity achievement}, where generativity refers to ``care and concern directed toward others, typically younger individuals,'' and covers elements such as desire and action to make a difference, to give back, to contribute to the lives of others or to create a legacy \citep[see][]{gruenewald2016BaltimoreExperienceCorps}. This outcome had an overall missingness of 22.6\%, and the missingness varies widely -- 10.1\% and 30.2\% in the high and low participation groups under treatment, and 26.3\% under control. Previous analyses used a mixture modeling method that assumes the missingness is MAR; we will relax this assumption and consider different missingness mechanisms under LMAR.

\section{PCE nonidentifiability re-examined}\label{sec:setting}

\subsection{Setting, notation and estimands}\label{subsec:notation}

Let $X$ denote baseline covariates, $Z$ denote treatment assignment (1 for the active treatment, 0 for the control condition) and $Y$ denote the observed outcome. Let $Y_z$ be the potential outcome had treatment $z$ been assigned ($z=0,1$). In defining these potential outcomes, we invoke the consistency assumption, which equates $Y=ZY_1+(1-Z)Y_0$.

Let $C$ denote compliance type. The person is called a complier ($C=1$) if they \textit{would}, and a noncomplier ($C=0$) if they \textit{would not}, comply to the active treatment \textit{if they were assigned to it} -- where what it means to ``comply'' to the active treatment is defined in a study-specific manner. 
% Note that $C$ itself is a potential outcome. 
In the treatment arm, $C$ is observed in the actual behavior of the person. In the control arm, the person is not offered the treatment, so $C$ is not observed.
Let $R$ be the binary \textit{response} indicator, where $R=1$ if $Y$ is observed, and $R=0$ if $Y$ is missing. 
The full data for an individual are $(X,C,Y_1,Y_0)$, and the observed data are $O:=(X,Z,ZC,R,RY)$. Assume that we observe $n$ i.i.d. copies of $O$.

The estimands of interest are the two PCEs (CACE and NACE, which we denote by $\Delta_c$, for $c=1,0$) and the ATE (which we denote by $\Delta$).
\begin{alignat*}{2}
    &\Delta_c
    &&:=\E[Y_1-Y_0\mid C=c],
    \\
    &\Delta
    &&:=\E[Y_1-Y_0].
\end{alignat*}

Below we will use additional shorthand notations, all of which will be introduced at first use. For ease of reference, these shorthand notations as well as the acronyms are listed in Table~\ref{tab:acronyms}.

\subsection{Base assumptions}

There are two base assumptions that we make throughout. They serve as starting points, and are not sufficient for identifying the effects in the absence of missing data or recovering identification in the presence of missing data.

\begin{center}
\begin{tabular}{lll}
    A0: 
    & $Z\independent (Y_1,Y_0,C)\mid X$ & (treatment assignment ignorability),~~and 
    \\
    & $0<\P(Z=1\mid X)<0$ & (treatment assignment positivity)
    \\
    B0: 
    & $R\independent Y\mid X,Z,C$ & (latent missing at random, or LMAR),~~and 
    \\
    & $\P(R=1\mid X,Z,C)>0$ & (response positivity)
\end{tabular}
\end{center}

Each of assumptions A0 and B0 above combines a conditional independence component with the related positivity requirement; the two components are always used together. To avoid cluttering presentation, we will often refer to A0 simply as treatment assignment ignorability and to B0 simply as LMAR, leaving the positivity requirement implicit.

For simplicity, we assume that A0 and B0 condition on the same set of baseline covariates, $X$. Things are more complicated if the two covariate sets differ, or if outcome missingness depends on variables that are post-treatment-assignment; this is left to future work.
Note that if treatment assignment is randomized, A0 holds with any set of baseline covariates.

Let us see what is gained by these base assumptions. Let $\delta_c(X):=\E[Y_1-Y_0\mid X,C=c]$ (for $c=1,0$) denote stratum-specific conditional effects.
It can be shown (see proof in the  Appendix) that the PCE for stratum $c$ is a weighted mean of $\delta_c(X)$,
\begin{align}
    \Delta_c
    &=\frac{\E[\delta_c(X)\P(C=c\mid X)]}{\E[\P(C=c\mid X)]}.\label{eq:Delta.c}
\end{align}
Under A0, the weighting function is identified,
\begin{align}
    \P(C=c\mid X)=\overbrace{\P(C=c\mid X,Z=1)}^{\textstyle=:\pi_c(X)},
\end{align}
and the conditional effect is equal to a difference in means (see proof in the  Appendix),
\begin{align}
    \delta_c(X)
    =\overbrace{\E[Y\mid X,C=c,Z=1]}^{\textstyle=:\mu_{1c}(X)}~-~\overbrace{\E[Y\mid X,C=c,Z=0]}^{\textstyle=:\mu_{0c}(X)}.\label{eq:delta.c}
\end{align}
This means that identification of $\Delta_c$ amounts to identifying $\mu_{1c}(X)$ and $\mu_{0c}(X)$. Identification of these two conditional means would also identify the ATE, which can be expressed as 
\begin{align}
    \Delta=\E[\delta_1(X)\P(C=1\mid X)+\delta_0(X)\P(C=0\mid X)].
\end{align}
Under B0, these means are equated to the means over the so-called complete cases,
\begin{align}
    \mu_{zc}(X)=\E[Y\mid X,C=c,Z=z,R=1].\label{eq:mu.complete.cases}
\end{align}
This identifies $\mu_{1c}(X)$ but not $\mu_{0c}(X)$ because $C$ is observed for $Z=1$ but not for $Z=0$. LMAR means MAR in the treatment arm and MNAR in the control arm.

To sum up, under the base assumptions, the PCEs and the ATE are unidentified due to the nonidentifiability of $\mu_{0c}(X)$.

%%%%%%%%%%%%%%%%%%%%%%%%%%%%%%%%%%%%%%%%%%%%%%%%%%%%%%%%%%%%%

\subsection{Nonidentifiability of $\mu_{0c}(X)$ under the base assumptions}\label{sec:nonidentifiability}

We seek to characterize this nonidentifiability of $\mu_{0c}(X):=\E[Y\mid X,C=c,Z=0]$ in precise terms, with the goal to obtain a general view that clarifies how it may be resolved. We start with the simpler case with full data before turning to the current case.

The reasoning below is all about the control condition ($Z=0$), and conditions on covariates throughout. This is explicit in notation, but is mostly implicit in verbal explanation.

\paragraph{The simpler case with outcome fully observed.}

In the control condition, the outcome distribution is a mixture of those from compliers and noncompliers, but it is not known who is a complier or a noncomplier. Leaning on the fact that the mixture mean, which is observed in this case, is a weighted average of component means,
\begin{align}
    \overbrace{\E[Y\mid X,Z=0]}^{\textstyle=:\kappa_0(X)}=
    \sum_{c=0}^1
    \overbrace{\E[Y\mid X,Z=0,C=c]}^{\textstyle=:{\color{red}\mu_{0c}(X)}}
    \underbrace{\P(C=c\mid X,Z=0)}_{\textstyle=\pi_c(X)\text{~\footnotesize under A0}},
    \label{eq:full.mixture.pre}
\end{align}
we have 
a mixture equation, i.e., an equation with two unknowns,
\begin{align}
    \kappa_0(X)={\color{red}\mu_{01}(X)}\pi_1(X)+{\color{red}\mu_{00}(X)}\pi_0(X).\label{eq:full.mixture}
\end{align}
This one equation characterizes the nonidentifiability of $\mu_{0c}(X)$ in this simple case with full data. In this case, any of the principal identification assumptions mentioned in Section~\ref{sec:intro} is sufficient to resolve this nonidentifiability and identify the PCEs.

\paragraph{Back to the current case.}

With outcome missingness, we have a similar equation to (\ref{eq:full.mixture.pre}), except where all the expectations and probabilities additionally condition on $R=1$ (observing the outcome).
\begin{align}
    \overbrace{\E[Y\mid X,Z=0,R=1]}^{\textstyle=:\kappa_0^R(X)}=
    \sum_{c=0}^1
    \underbrace{\E[Y\mid X,Z=0,C=c,R=1]}_{\textstyle={\color{red}\mu_{0c}(X)}\text{~\footnotesize under B0 (see (\ref{eq:mu.complete.cases}))}}
    \overbrace{\P(C=c\mid X,Z=0,R=1)}^{\textstyle=:{\color{violet}\pi_{0c}^R(X)}}.\label{eq:outcome.mixture.pre}
\end{align}
This gives us the \textit{outcome mixture equation}
\begin{align}
    \kappa_0^R(X)={\color{red}\mu_{01}(X)}{\color{violet}\pi_{01}^R(X)}+{\color{red}\mu_{00}(X)}{\color{violet}\pi_{00}^R(X)},\label{eq:outcome.mixture}
\end{align}
which is similar to (\ref{eq:full.mixture}), but with a different term on the left hand side and a different pair of mixture weights.
Under A0, it can be shown (see proof in the  Appendix) that
\begin{align}
    \pi_{0c}^R(X)=
    \pi_c(X)
    \overbrace{\P(R=1\mid X,Z=0,C=c)}^{\textstyle=:\varpi_{0c}(X)}
    /
    \overbrace{\P(R=1\mid X,Z=0)}^{\textstyle=:\lambda_0(X)},\label{eq:mix.wts}
\end{align}
so we can also write (\ref{eq:outcome.mixture}) as
\begin{align}
    \kappa_0^R(X)=
    {\color{red}\mu_{01}(X)}\underbrace{\frac{\pi_1(X){\color{purple}\varpi_{01}(X)}}{\lambda_0(X)}}_{\textstyle{\color{violet}\pi_{01}^R(X)}}
    \,+\,
    {\color{red}\mu_{00}(X)}\underbrace{\frac{\pi_0(X){\color{purple}\varpi_{00}(X)}}{\lambda_0(X)}}_{\textstyle{\color{violet}\pi_{00}^R(X)}}.\tag{\ref{eq:outcome.mixture}}
\end{align}
% (\ref{eq:outcome.mixture.long}) is more complicated than (\ref{eq:full.mixture}) because the mixture weights involve the stratum-specific response probabilities $\varpi_{0c}(X)$ that are not identified under LMAR.

In addition to this outcome mixture equation, we have a \textit{response mixture equation}, obtained by treating $R$ as an outcome. This equation is of the same simple form as (\ref{eq:full.mixture}),
\begin{align}
    \lambda_0(X)={\color{purple}\varpi_{01}(X)}\pi_1(X)+{\color{purple}\varpi_{00}(X)}\pi_0(X).\label{eq:response.mixture}
\end{align}

To sum up, under the combination of A0 and B0, the PCE nonidentifiability problem boils down to the pair of mixture equations -- (\ref{eq:outcome.mixture}) obtained under A0 and B0 and (\ref{eq:response.mixture}) obtained under A0 -- with a combination of four unknowns.

\section{A general solution to the nonidentifiability problem and implications for practice}\label{sec:solutions}

To identify $\mu_{0c}(X)$, we need to find a way to solve the response mixture equation (\ref{eq:response.mixture}), and then plug the solution into the outcome mixture equation (\ref{eq:outcome.mixture}) and find a way solve for $\mu_{0c}(X)$. The first step requires adding an assumption about the missingness that is more specific than LMAR. The second step requires a principal identification assumption; this assumption is required regardless of whether there is missing data.

\subsection{Identifying the mixture weights assuming a specific missingness mechanism}\label{sec:specific-miss-assumptions}

To solve (\ref{eq:response.mixture}) one might consider different assumptions, including but not limited to those listed here. We roughly group these assumptions in two categories: those that compare missingness across treatment conditions (including assumptions that already exist), and those that compare missingness across principal strata (a new category we propose).
All these assumptions condition on baseline covariates.
For simplicity we assume that the set of covariates is the same as in the base assumptions.

\subsubsection{Assumptions comparing missingness across treatment conditions}\label{sec:specific-miss-type1}

Recall that the two unknowns in (\ref{eq:response.mixture}) are $\color{purple}{\varpi_{01}(X)}$ and $\color{purple}{\varpi_{00}(X)}$, the conditional response probabilities under control among compliers and noncompliers. Any assumption that pins down one of these unknowns by relating it to missingness probabilities in the intervention arm will help solve (\ref{eq:response.mixture}) and identify both probabilities (and thus the mixture weights).
The two specific missingness assumptions that have appeared in the literature, SNR/rER and SCR (see Section~\ref{sec:intro}), both belong in this category. We will start with these and give a sense of considerations that might suggest either of them. Then, and this is important, we will point out an under-appreciated limitation of these assumptions, and propose that they be replaced with a modified version that avoids this limitation. We will also comment on the general category of assumptions that compare missingness across treatment conditions.

% Within this category, we will first present the stable response assumptions mentioned in Section \ref{sec:intro} and the intuition behind them. We will then point out an under-appreciated limitation of these assumptions and, for an ad hoc fix, propose an approximate variant of these assumptions that avoids the limitation.

\paragraph{Stable response of a stratum.}
These assumptions can be stated in in a simple form:

\begin{center}
\begin{tabular}{ll}
    SNR/rER: & $R\independent Z\mid X,C=0$
    \\
    SCR: & $R\independent Z\mid X,C=1$
\end{tabular}
\end{center}

\noindent These assumptions say that for one stratum (noncompliers in SNR/rER and compliers in SCR), within covariate levels the outcome is missing at the same rate in the two arms of the study.
% These assumptions are thus mirror images of each other. 
Each assumption obtains one of the unknowns in (\ref{eq:response.mixture}), allowing it to be solved for the other one.
SNR/rER equates ${\color{purple}\varpi_{00}(X)}=\varpi_{10}(X)$; SCR equates ${\color{purple}\varpi_{01}(X)}=\varpi_{11}(X)$.
% SNR/rER equates ${\color{purple}\varpi_{00}(X)}=\varpi_{10}(X)$ and solves (\ref{eq:response.mixture}) for ${\color{purple}\varpi_{01}(X)}$. SCR equates ${\color{purple}\varpi_{01}(X)}=\varpi_{11}(X)$ and solves (\ref{eq:response.mixture}) for ${\color{purple}\varpi_{00}(X)}$. This then identifies the mixture weights ${\color{violet}\pi_{0c}^R(X)}$.

These two assumptions require different sorts of justification. SNR/rER may be invoked based on the idea that treatment assignment acts as an instrument vis-a-vis response (treated as a second outcome). This should be justified separately from, say, an assumption that treatment assignment is an instrument with respect to the outcome, because even if the latter assumption holds, that does not imply that treatment assignment operates the same way with respect to response. SNR/rER is violated, for example, if being assigned to treatment makes noncompliers feel more cared for and thus are more likely to provide outcome assessment. 
% Each assumption should stand on its own.

SCR, on the other hand, might be motivated by the idea that compliers are people with a strong tendency to do what they are asked to do, so their response to outcome assessment -- a form of ``compliance'' -- should be the same regardless of study arm. This is a strong assumption, as arguably compliance/response may be specific to the task (volunteering vs. responding to outcome assessment) and the context (assignment to intervention vs. control).

These two different assumptions are not to be combined, not only because they are motivated differently, but also because their combination would imply that conditional on covariates (but not compliance type) missing rates are equal between the intervention and control conditions, which may well contradict what is observed.

\paragraph{A limitation of stable response.}

The above two assumptions unfortunately can imply implausible response probabilities that are negative or greater than 1. 
This does not mean that they necessarily do in any application, but the possibility should cause concern. Take rER/SNR for example. It implies that
% \vspace{-1em}
% \begin{align}
%     {\color{purple}\varpi_{00}(X)}=\varpi_{10}(X),~~~{\color{purple}\varpi_{01}(X)}=\frac{\lambda_0(X)-\pi_0(X)\varpi_{10}(X)}{\pi_1(X)}.
% \end{align}
% \vspace{-3em}
${\color{purple}\varpi_{01}(X)}=\frac{\lambda_0(X)-\pi_0(X)\varpi_{10}(X)}{\pi_1(X)}$. 
If for some values of $X$, the response probability of noncompliers in the \textit{intervention} arm $\varpi_{10}(X)$ exceeds $\lambda_0(X)/\pi_0(X)$ (where recall that $\lambda_0(X)$ is the response probability in the \textit{control} arm unconditional on compliance type and $\pi_0(X)$ is the noncomplier prevalence), then \mbox{${\color{purple}\varpi_{01}(X)}<0$}. If for some other $X$ values, $\varpi_{10}(X)$ is smaller than $[\lambda_0(X)-\pi_1(X)]/\pi_0(X)$, then ${\color{purple}{\varpi_{01}(X)}}>1$.
The same issue affects $\color{purple}\varpi_{00}(X)$ under SCR, which is a miror image of SNR.

% \textit{A drawback of both stable response assumptions}, which to the best of our knowledge has not been pointed out in the literature, is that they can, for some covariate values, imply implausible response probabilities that are negative or greater than 1. It can be seen in the formulas for $\varpi_{01}(X)$ under SNR/rER and $\varpi_{00}(X)$ under SCR in Table~\ref{tab:missingprobs} that there is no guarantee for these quantities to be in range. When they are not, the assumption conflicts with the observed data distribution. We will comment further on this issue in Section \ref{sec:rER}.

% The identification results under B1a and B2a are simple (see Table \ref{tab:missingprobs}), but one needs to watch out for an anomaly: it is possible for one stratum's response probability (the one that results from solving (\ref{eq:response.mixture})) to be out of range (for some $X$ values). That probability may exceed 1 if the missingness is much more severe in the control than in the treatment arm, and may be negative if the opposite is true. This is more likely to happen if this stratum is the smaller stratum (conditional on $X$).

While there now exists a literature following and extending on the original rER assumption \citep[e.g.,][and work cited in Section~\ref{sec:intro}]{dunn2005EstimatingTreatmentEffects,zhou2006ITTAnalysisRandomized,taylor2009MultipleImputationMethods,chen2009IdentifiabilityEstimationCausal,lui2010NotesOddsRatio,mealli2012RefreshingAccountPrincipal,chen2015SemiparametricInferenceComplier}, 
surprisingly, 
% the issue that this assumption can imply
this risk of
implausible response probabilities seems to be under-appreciated. This might be because previous work, when deriving $\color{red}\mu_{0c}(X)$, did this only for one setting (compound ER), and thus did not need to derive (let alone scrutinize) the response probabilities $\color{purple}\varpi_{0c}(X)$ as an intermediate step (see more details on this in the Appendix).

The consequence of any implausible response probabilities (a contradiction with the observed data distribution) is that the mixture weights $\color{violet}{\pi_{01}^R(X)}$ and $\color{violet}{\pi_{00}^R(X)}$ they imply are unrealistic.
% (in serious cases one may be negative and the other greater than 1). 
This leads to unrealistic $\color{red}{\mu_{01}(X)}$ or $\color{red}{\mu_{00}(X)}$ or both, depending on the principal identification assumption.

% We draw attention to this because if for some $X$ values, SNR/rER dictates ${\color{purple}\varpi_{01}(X)}$ (or SCR dictates ${\color{purple}\varpi_{00}(X)}$) out of range, the implied mixture weights ${\color{violet}\pi_{01}^R(X)}$ and ${\color{violet}\pi_{00}^R(X)}$ are unrealistic (in serious cases one may be negative and the other greater than 1). This leads to unrealistic ${\color{red}\mu_{01}(X)}$ or ${\color{red}\mu_{00}(X)}$ or both, depending on the principal identification assumption.

\paragraph{An improvement: switching to \textit{near stable} response.}

When substantive considerations point toward a stable response type of assumption, it is important to have some guard against a possible case of implausible response probabilities. There are multiple ways this can potentially be tackled. As this is a problem of stable response assumptions specifically and is not central to the objective of developing a general effect identification template, we propose a simple ad hoc fix: a modification of these assumptions to avoid this problem.

% This is a form of contradiction with the observed data distribution, so one option is to estimate the implied response probabilities based on data and check the degree to which they fall out of bounds, as a way to decide whether the assumption is problematic. As this 

% The above-mentioned problem, when it occurs, is a conflict of the assumption with the observed data distribution.
We call the modified assumptions near-SNR and near-SCR, and number them B1 and B2. These assumptions are close to SNR and SCR but respect the legal range of the response probabilities. Near-SNR coincides with SNR at $X$ values where SNR implies a \textit{proper} response probability under control for compliers (defined as ${\color{purple}\varpi_{01}(X)}\in[\epsilon,1]$ for some small $\epsilon>0$). For $X$ values where SNR implies that this probability is below $\epsilon$ (or above 1), near-SNR assumes the probability is $\epsilon$ (or 1). Near-SCR is similarly formulated, and differs from SCR at $X$ values where SCR implies an improper value for the noncomplier response probability under control $\color{purple}\varpi_{00}(X)$. The use of $\epsilon$ instead of zero for the lower end of the response probability range is to respect the positivity component of assumption B0.

This modification from \textit{stable} to \textit{near stable} response is simply a shift from an exact to an approximate version of an assumption, which is likely fine in practice because substantive considerations are often rough rather than exact. The two versions coincide when the exact version does not contradict the observed data distribution, and the approximate version avoids such conflict all together. Therefore we propose, if substantive considerations suggest a stable response assumption, that the researcher adopt the near stable version instead.

% We propose, if substantive considerations (which are often rough in practice) suggest a stable response assumption, that the researcher adopt the near stable version instead. If the exact stable response assumption does not conflict with the observed data distribution, the two versions coincide. If the exact assumption conflicts with the observed data distribution, then it is better to use the near stable response version which avoids this conflict.

Identification results for the response probabilities and mixture weights under near-SNR and near-SCR are included in Table~\ref{tab:missingprobs}, and all proofs are provided in the  Appendix.

\smallskip

\paragraph{A general note about this category of assumptions.}

Above we have considered only assumptions of the stable response flavor, but any assumption that belong in this category of comparing missingness \textit{across treatment conditions} (as a way to obtain one of the two unknowns in the mixture equation (\ref{eq:response.mixture})) can have this problem of implying implausible response probabilities. The same strategy of switching to an approximate version applies.

\begin{table}[t!]
    \caption{Identification of the stratum-specific conditional probabilities of observing the outcome under control (${\color{purple}\varpi_{01}(X)}$, ${\color{purple}\varpi_{00}(X)}$) and of the mixture weights (${\color{violet}\pi_{01}^R(X)}$, ${\color{violet}\pi_{00}^R(X)}$), under different \textit{specific missingness mechanisms}}
    \label{tab:missingprobs}
    \centering
    \resizebox{\textwidth}{!}{%
    \begin{tabular}{lcccc}
         \textit{Missingness} & \multicolumn{4}{c}{\textit{Identification results for}}
         \\\cline{2-5}
         \\[-.9em]
         \textit{assumption}
         & ${\color{purple}\varpi_{01}(X)}$ & ${\color{purple}\varpi_{00}(X)}$ 
         & ${\color{violet}\pi_{01}^R(X)}$ & ${\color{violet}\pi_{00}^R(X)}$
         \\\toprule
         \multicolumn{5}{l}{Assumptions comparing missingness across treatment conditions:}
         \\[.8em]
         % SNR/rER 
         % & $\frac{\lambda_0(X)-\pi_0(X)\varpi_{10}(X)}{\pi_1(X)}$ 
         % & $\varpi_{10}(X)$ 
         % && $1-{\color{violet}\pi_{00}^R(X)}$ 
         % & $\pi_0(X)\frac{\varpi_{10}(X)}{\lambda_0(X)}$
         % \\[.8em]
         % SCR
         % & $\varpi_{11}(X)$ 
         % & $\frac{\lambda_0(X)-\pi_1(X)\varpi_{11}(X)}{\pi_0(X)}$ 
         % && ~~~$\pi_1(X)\frac{\varpi_{11}(X)}{\lambda_0(X)}$~~~
         % & $1-{\color{violet}\pi_{01}^R(X)}$
         % \\[.8em]
         B1: near-SNR
         & {\footnotesize\begin{tabular}{@{}c@{}}
             $\frac{\lambda_0(X)-\pi_0(X)\varpi_{10}(X)}{\pi_1(X)}$ if within range, 
             \\
             else restrict to range
         \end{tabular}}
         & $\frac{\lambda_0(X)-\pi_1(X){\color{purple}\varpi_{01}(X)}}{\pi_0(X)}$
         & $\pi_1(X)\frac{{\color{purple}\varpi_{01}(X)}}{\lambda_0(X)}$
         & $1-{\color{violet}\pi_{01}^R(X)}$
         \\[.8em]
         B2: near-SCR
         & $\frac{\lambda_0(X)-\pi_0(X){\color{purple}\varpi_{00}(X)}}{\pi_1(X)}$ 
         & {\footnotesize\begin{tabular}{@{}c@{}}
             $\frac{\lambda_0(X)-\pi_1(X)\varpi_{11}(X)}{\pi_0(X)}$ if within range, 
             \\
             else restrict to range
         \end{tabular}} 
         & $1-{\color{violet}\pi_{00}^R(X)}$
         & $\pi_0(X)\frac{{\color{purple}\varpi_{00}(X)}}{\lambda_0(X)}$
         \\[.8em]
         \midrule
         \multicolumn{5}{l}{Assumptions comparing missingness across principal strata:}
         \\[.5em]
         B3: rPI
         & $\lambda_0(X)$ & $\lambda_0(X)$
         & $\pi_1(X)$ & $\pi_0(X)$
         \\[.8em]
         % rPP
         % & \multicolumn{2}{c}{$\varpi_{1c}(X)\frac{\lambda_0(X)}{\varpi_1(X)}$}
         % && \multicolumn{2}{c}{$\pi_c^{R1}(X)$}
         % \\[.8em]
         B4: rPO
         & $\frac{\gamma_1(X)-\omega_1(X)}{2[\varrho_1(X)-1]\pi_1(X)}$ 
         & $\frac{\gamma_0(X)-\omega_0(X)}{2[\varrho_0(X)-1]\pi_0(X)}$
         & $\frac{\gamma_1(X)-\omega_1(X)}{2[\varrho_1(X)-1]\lambda_0(X)}$
         &  $\frac{\gamma_0(X)-\omega_0(X)}{2[\varrho_0(X)-1]\lambda_0(X)}$
         \\[-.8em]
         \\\bottomrule
         \\[-.8em]
         \multicolumn{5}{l}{\footnotesize
         \begin{tabular}{@{}l@{}l@{}}
             Notes:~
             & $\varrho_c(X):=\frac{\varpi_{1c}(X)/[1-\varpi_{1c}(X)]}{\varpi_{1(1\!-\!c)}(X)/[1-\varpi_{1(1\!-\!c)}(X)]}$ is the complier-to-noncomplier odds ratio of response in the treatment arm.
             \\
             & $\gamma_c(X):=[\pi_c(X)+\lambda_0(X)][\varrho_c(X)-1]+1$.
             ~$\omega_c(X):=\sqrt{[\gamma_c(X)]^2-4\pi_c(X)\lambda_0(X)\varrho_c(X)[\varrho_c(X)-1]}$.
         \end{tabular}}
    \end{tabular}%
    }
\end{table}

\subsubsection{Assumptions comparing missingness across principal strata}

We propose to add this second type of assumptions to the researcher's toolbox to expand options.
Assumptions of this type are specifically about how the missingness under control compares across the principal strata, which helps solve equation (\ref{eq:response.mixture}).
We list here two assumptions as examples; these assumptions do not have the implausible response probabilities issue. Results under these assumptions are included in Table~\ref{tab:missingprobs}.

\paragraph{Response PI.}
This assumption is of the same flavor as the PI assumption mentioned in the literature review, except it is with respect to response (i.e., observing the outcome) rather than the outcome. We reserve the simple label PI for the principal identification assumption PI, and refer to the current missingness assumption as \textit{response PI} (rPI).

% This assumption is not to be confused with the principal identification assumption with the same name. We will refer to that assumption (which is principal ignorability w.r.t. the potential outcome $Y_0$) simply as PI, and the current assumption as PI for response (rPI).
\begin{center}
    B3 (rPI):~~~~~$R\independent C\mid X,Z=0$
\end{center}

% rPI says conditional on covariates compliers and noncompliers have the same response probabilities under control,

% Unlike B1 and B2, B3 is only concerned with the control arm and is agnostic of what happens in the treatment arm.
rPI says that within levels of covariates, compliers and noncompliers share the same response probability under control, i.e., ${\color{purple}\varpi_{01}(X)}={\color{purple}\varpi_{00}(X)}$.
% specifically ${\color{purple}\varpi_{01}(X)}={\color{purple}\varpi_{00}(X)}$, 
This combined with (\ref{eq:response.mixture}) identifies both probabilities by $\lambda_0(X)$, and thus identifies the mixture weights as ${\color{violet}\pi_{0c}^R(X)}=\pi_c(X)$.

A different argument -- one that yields the same result -- is that rPI combined with LMAR implies that the missingness under control is MAR, $R\independent Y\mid X,Z=0$ (see proof in the  Appendix). This means $\kappa_0^R(X)=\kappa_0(X)$, which justifies using the full-data outcome mixture equation (\ref{eq:full.mixture}) (which has mixture weights $\pi_c(X)$) but replacing $\kappa_0(X)$ with $\kappa_0^R(X)$.

% It is note-worthy that the combination of LMAR with either type of principal ignorability -- principal ignorability for outcome (PI, see section \ref{subsec:PI}) or principal ignorability for response (rPI, here) -- implies MAR.

To justify rPI one needs to be confident that the covariates $X$ capture important charactertics and contextual factors that make a person more or less likely to respond to outcome assessment under control, such that that behavior is no longer associated with whether that person is a complier or noncomplier.

\paragraph{Proportional response odds.}
This assumption, unlike rPI, allows compliers and noncompliers to differ in outcome missingness under control. It posits that this difference, in the specific form of the complier-to-noncomplier odds ratio of response, is shared between the two study arms.

\begin{center}
    % B3a (rPP): & $\displaystyle\frac{\varpi_{01}(X)}{\varpi_{00}(X)}=\frac{\varpi_{11}(X)}{\varpi_{10}(X)}$
    % \\[1em]
    B4 (rPO): 
    % & $\displaystyle\frac{\varpi_{01}(X)/[1-\varpi_{01}(X)]}{\varpi_{00}(X)/[1-\varpi_{00}(X)]}=\frac{\varpi_{11}(X)/[1-\varpi_{11}(X)]}{\varpi_{10}(X)/[1-\varpi_{10}(X)]}$, 
    % \\[1.5em]
    ~~~~$\displaystyle\frac{\text{odds}(R=1\mid X,C=1,Z=0)}{\text{odds}(R=1\mid X,C=0,Z=0)}=\frac{\text{odds}(R=1\mid X,C=1,Z=1)}{\text{odds}(R=1\mid X,C=0,Z=1)}$
\end{center}

Unlike rPI, which is agnostic of what happens in the treatment arm, rPO uses what happens in the treatment arm to help infer response probabilities under control. However, unlike SNR and SCR, rPO does not treat compliers or noncompliers differently.

rPO could be understood as a no-interaction assumption. Specifically, the assigned treatment condition $Z$ does not interact with compliance type and covariates on the odds of missingness; its effect on the odds of missingness is separate from the effects of covariates and compliance type. No-interaction assumptions have been used elsewhere \citep[e.g.,][]{joffe2008ExtendedInstrumentalVariables,gallop2009MediationAnalysisPrincipal}.

\medskip

Above are just several specific missingness assumptions. Many other may be considered. For example, one might find the choice of the odds ratio as the invariant form in rPO arbitrary, and may instead choose to use, say, a probability ratio. Then there are two choices, either a proportional response probability assumption or a proportional missingness probability assumption (because probability ratios, unlike odds ratios, are asymmetric).
Another case is where one believes that compliers and noncompliers do differ in outcome missingness but they differ less in the control arm than in the treatment arm. In this case, it may be appropriate to choose an assumption that is in between rPI and rPO.

% While rPI equates ${\color{purple}\varpi_{01}(X)}$ and ${\color{purple}\varpi_{00}(X)}$, B3 let these response probabilities differ, but have them informed by how complier and noncomplier response probabilities differ in the treatment condition. B3 can also be seen as a middle point between SNR/rER and SCR.

% Under proportional response probability (B3a) the mixture weights ${\color{violet}\pi_c^*(X)}$, which are the conditional principal stratum prevalences among complete cases in the control condition, are identified by the corresponding prevalences in the treatment condition, $\pi_c^{R1}(X)$. Under proportional response odds (B3b), the identification result for the mixture weights is slightly more complicated.

\subsection{Identifying $\mu_{0c}(X)$ based on a principal identification assumption (and the mixture weights)}\label{sec:principalassumptions}

We now use principal identification assumptions to solve the outcome mixture equation (\ref{eq:outcome.mixture}), contingent on the mixture weights $\color{violet}{\pi_{0c}^R(X)}$ being identified, for the most part. We show this for the assumptions mentioned in Section~\ref{sec:intro}, but any assumption that would point-identify $\color{red}\mu_{0c}(X)$ in the full-data setting suffices.

We state these assumptions in a manner that focuses on the important element. Specifically, where multiple (e.g., deterministic and stochastic) versions of an assumption would work, we use the weaker one (the stochastic version). Also, while some of these assumptions are stated in the literature as involving potential outcomes, we translate them into a version without potential outcomes; this is motivated by the fact that we are using the assumption to identify ${\color{red}\mu_{0c}(X)}:=\E[Y\mid X,Z=0,C=c]$, a conditional mean of $Y$, not of a potential outcome. The base assumption A0 has finished the job of translating from potential to actual outcomes, so we need not return to potential outcomes at this point.

% We derive $\color{red}{\mu_{0c}(X)}$ under each of these assumptions combined with the base assumptions, contingent on the mixture weights ${\color{violet}\pi_{0c}^R(X)}$ being identified, for the most part. 
Results under these assumptions appear in Table~\ref{tab:identification}, and proofs are in the  Appendix.

% and provide identification results for the missing outcome setting -- contingent on identification of the mixture weights ${\color{violet}\pi_{01}^R(X)}$ and ${\color{violet}\pi_{00}^R(X)}$ (see section~\ref{sec:specific-miss-assumptions}).

% Much has been written about different identification approaches, which we will not repeat here. Rather, we just state the assumptions and the corresponding identification results.

\paragraph{ER assumption.}
% At the heart of the IV approach, the ER assumption says that treatment assignment does not affect the outcome other than through treatment received. 
% It implies that the NACE is zero, and the CACE explains the full effect of treatment assignment.
% it is questionable if treatment received is not strictly binary \citep{marshall2016CoarseningBiasHow,andresen2021InstrumentbasedEstimationBinarised}. 
% The ER assumption has several versions. The conditional independence (or mean independence) version below suffices.

For the current one-sided noncompliance setting, ER can be stated as the conditional independence or conditional mean independence below.

\begin{center}
    A1 (ER):~~~$Y\independent Z\mid X,C=0$~~or~~$\overbrace{\E[Y\mid X,Z=0,C=0]}^{\textstyle=:{\color{red}\mu_{00}(X)}}=\overbrace{\E[Y\mid X,Z=1,C=0]}^{\textstyle=:\mu_{10}(X)}$.
\end{center}

ER identifies ${\color{red}\mu_{00}(X)}=\mu_{10}(X)$, so one solves (\ref{eq:outcome.mixture}) for ${\color{red}\mu_{01}(X)}$. The result (see Table~\ref{tab:identification}) is a simple translation of the full-data setting result, replacing $\kappa_0(X)$ with $\kappa_0^R(X)$ and replacing $\pi_1(X)$ with $\pi_{01}^R(X)$.

Note that ER implies that the CACE explains the full average treatment effect, while the NACE is zero, and it is zero within any group defined based on $X$. This is a strong assumption. 
% that should be carefully justified. 
ER is suspect
% for example, 
if noncompliers in the treatment arm do experience some active intervention elements \citep{marshall2016CoarseningBiasHow,andresen2021InstrumentbasedEstimationBinarised}. ER also does not allow for any psychological effects or effects of compensatory behaviors. 

% Similar to SNR and SCR, it is possible for ER to imply ${\color{red}\mu_{01}(X)}$ values outside the range of a bounded outcome. 
% % One might consider a modification similar to SNR2 to keep ${\color{red}\mu_{01}(X)}$ values within range. This is labeled \textit{near ER} in Table~\ref{tab:identification}.
% We will not address this issue here, as it is off the paper's main topic.
% % This has been discussed elsewhere [CITES].

\begin{table}[t!]
    \caption{Identification of the stratum-specific conditional outcome means under control (${\color{red}\mu_{01}(X)}$ and ${\color{red}\mu_{00}(X)}$) under varying principal identification assumptions, contingent on identification of the mixture weights (${\color{violet}\pi_{01}^R(X)}$ and ${\color{violet}\pi_{00}^R(X)}$)}
    \label{tab:identification}
    \centering
    \resizebox{.95\textwidth}{!}{%
    \begin{tabular}{lcccc}
        \textit{Principal identification}  
        && \multicolumn{3}{c}{\textit{Identification results for}} 
        \\\cline{3-5}
        \textit{assumption} 
        && ${\color{red}\mu_{01}(X)}$
        && ${\color{red}\mu_{00}(X)}$
        \\\toprule
        \\[-.7em]
        A1: ER  
        && $\mu_{10}(X)+\mfrac{\kappa_0^R(X)-\mu_{10}(X)}{{\color{violet}\pi_{01}^R(X)}}$ 
        && ~~~~~~~~~~$\mu_{10}(X)$~~~~~~~~~~
        % \\[1em]
        % near-ER  
        % && {\footnotesize\begin{tabular}{@{}c@{}}
        %     same as above if within outcome bounds,
        %     \\
        %     otherwise restrict to bounds
        % \end{tabular}}
        % && ${\color{red}\mu_{01}(X)}+\mfrac{\kappa_0^R(X)-{\color{red}\mu_{01}(X)}}{{\color{violet}\pi_{00}^R(X)}}$
        \\[-.5em]
        \\\hline
        \\[-.5em]
        A2: PI 
        && $\kappa_0^R(X)$ 
        && $\kappa_0^R(X)$
        \\[-.5em]
        \\\hline
        \\[-.5em]
        % PIsens-OR & a LMAR+ assumption
        % && \multicolumn{2}{c}{$\displaystyle\frac{\alpha_c(X)-\sqrt{[\alpha_c(X)]^2-4{\color{violet}\pi_c^*(X)}\kappa_0^R(X)\psi_c(\psi_c-1)}}{2(\psi_c-1){\color{violet}\pi_c^*(X)}}$}
        % \\[1.2em]
        A3a: PIsens-GOR 
        && $\mfrac{\alpha_{1,{\color{violet}\pi},\kappa}(X)-\beta_{1,{\color{violet}\pi},\kappa}(X)}{2(\psi_1-1){\color{violet}\pi_{01}^R(X)}}(h\!-\!l)\!+\!l$
        && $\mfrac{\alpha_{0,{\color{violet}\pi},\kappa}(X)-\beta_{0,{\color{violet}\pi},\kappa}(X)}{2(\psi_0-1){\color{violet}\pi_{00}^R(X)}}(h\!-\!l)\!+\!l$
        \\[1.2em]
        A3b: PIsens-MR 
        && $\mfrac{\rho_1\kappa_0^R(X)}{(\rho_1-1){\color{violet}\pi_{01}^R(X)}+1}$
        && $\mfrac{\rho_0\kappa_0^R(X)}{(\rho_0-1){\color{violet}\pi_{00}^R(X)}+1}$
        \\[1.2em]
        A3c: PIsens-SMDe 
        && $\kappa_0^R(X)+\eta\mfrac{{\color{violet}\pi_{00}^R(X)}\varsigma_0^R(X)}{\sqrt{1+\eta^2{\color{violet}\pi_{01}^R(X)}{\color{violet}\pi_{00}^R(X)}}}$
        && $\kappa_0^R(X)-\eta\mfrac{{\color{violet}\pi_{01}^R(X)}\varsigma_0^R(X)}{\sqrt{1+\eta^2{\color{violet}\pi_{01}^R(X)}{\color{violet}\pi_{00}^R(X)}}}$
        % \\[-.8em]
        % \\\hline
        % \\[-.5em]
        % ER-PI interpolation & LMAR
        % && unidentified & unidentified
        % \\[.8em]
        % ER-PI interpolation & a LMAR+ assumption
        % && $\displaystyle\frac{\kappa_0^R(X)-\omega{\color{violet}\pi_0^*(X)}\mu_{10}(X)}{1-\omega{\color{violet}\pi_0^*(X)}}$
        % & $\displaystyle\frac{(1-\omega)\kappa_0^R(X)+\omega{\color{violet}\pi_1^*(X)}\mu_{10}(X)}{1-\omega{\color{violet}\pi_0^*(X)}}$
        \\[-.8em]
        \\\bottomrule
        \\[-.8em]
        \multicolumn{5}{l}{\footnotesize
        \begin{tabular}{@{}l@{}l}
            Notes:~ 
            & $\psi_c=\psi^{2c-1}$. $\rho_c=\rho^{2c-1}$. 
            ~~$\varsigma_0^R(X):=\sqrt{\varsigma_0^{2R}(X)}$ where $\varsigma_0^{2R}(X):=\var(Y\mid X,Z=0,R=1)$.
            \\
            & $\alpha_{c,{\color{violet}\pi},\kappa}(X):=[{\color{violet}\pi_{0c}^R(X)}+\frac{\kappa_0^R(X)-l}{h-l}](\psi_c-1)+1$. 
            ~~$\beta_{c,{\color{violet}\pi},\kappa}(X):=\Big\{[\alpha_{c,{\color{violet}\pi},\kappa}(X)]^2-4{\color{violet}\pi_{0c}^R(X)}\frac{\kappa_0^R(X)-l}{h-l}\psi_c(\psi_c-1)\Big\}^{1/2}$.
        \end{tabular}
        }
    \end{tabular}%
    }
\end{table}

\paragraph{PI assumption.}
The PI assumption is usually stated as a conditional independence or conditional mean independence between a potential outcome and compliance type.
It is, however, sufficient to use a derived form of PI that does not involve potential outcomes:

\begin{center}
    A2 (PI):~~~$C\independent Y\mid X,Z=0$,~~~or~~~$\overbrace{\E[Y\mid X,Z=0,C=1]}^{\textstyle=:{\color{red}\mu_{01}(X)}}=\overbrace{\E[Y\mid X,Z=0,C=0]}^{\textstyle=:{\color{red}\mu_{00}(X)}}$.
\end{center}

PI reduces the number of unknowns in equation (\ref{eq:outcome.mixture}) to one, immediately yielding the solution ${\color{red}\mu_{01}(X)}={\color{red}\mu_{00}(X)}=\kappa_0^R(X)$.
This is a special case where the base assumptions combined with one additional assumption (PI) identifies $\mu_{0c}(X)$ (and thus the PCEs and ATE), obviating the need to identify the mixture weights.
Another interesting point about PI is that combined with LMAR it implies that the outcome is MAR, i.e., $R\independent Y\mid X,Z=0$. 
It is note-worthy that the combination of LMAR with either rPI (see section \ref{sec:specific-miss-assumptions}) or PI (here) implies MAR.

To justify PI one needs to be confident that the covariates $X$ capture the common causes of compliance to treatment and of the outcome under control. This is a strong assumption. PI may be appealing, however, in studies that collect rich baseline data.

\paragraph{Sensitivity assumptions deviating from PI.}
As the primary assumptions ER and PI are untestable, if either assumption is used, some sensitivity analysis should ideally be done.
Several sensitivity analyses for PI violation have been developed \citep{ding2017PrincipalStratificationAnalysis,nguyen2023SensitivityAnalysisPrincipal}, each replacing PI with a sensitivity assumption that deviates from PI. The assumptions are parameterized with different sensitivity parameters (odds ratio (OR), generalized odds ratio (GOR), mean ratio (MR) and standardized mean difference (SMD)) suitable for different outcome types. In these prior works, these assumptions are stated as assumptions about $Y_0$, but here we read them as assumptions about $Y$ under control.

\begin{center}
\begin{tabular}{ll}
    A3a (PIsens-GOR): 
    & $\displaystyle\frac{[\mu_{01}(X)-l]/[h-\mu_{00}(X)]}{[\mu_{00}(X)-l]/[h-\mu_{00}(X)]}=\psi$ for a range of $\psi$ considered plausible, 
    \\[.8em]
    & ~where $l,h$ are the lower and upper bounds of the outcome.
    \\[.5em]
    A3b (PIsens-MR):
    & $\displaystyle\frac{\mu_{01}(X)}{\mu_{00}(X)}=\rho$ for a range of $\rho$ considered plausible.
    \\[1em]
    A3c (PIsens-SMDe):
    & $\displaystyle\frac{\mu_{01}(X)-\mu_{00}(X)}{\sqrt{[\sigma_{01}^2(X)+\sigma_{00}^2(X)]/2}}=\eta$ for a plausible range of $\eta$, and
    \\[.8em]
    & $\sigma_{01}^2(X)=\sigma_{00}^2(X)$, where $\sigma_{0c}^2(X)\!:=\!\var(Y\!\mid\! X,Z\!=\!0,C\!=\!c)$, for $c\!=\!1,0$.
\end{tabular}
\end{center}

\noindent The OR-based assumption for a binary outcome is covered in A3a, with $l=0,h=1$.

All of these assumptions tie $\color{red}\mu_{01}(X)$ and $\color{red}\mu_{00}(X)$ together, thus also effectively reducing the number of unknowns in equation (\ref{eq:outcome.mixture}) to one. The solutions are in Table~\ref{tab:identification}.

\subsection{A template for practice}\label{sec:considerations}

Above we use one assumption to solve the response mixture equation then use another to solve the outcome mixture equation. In practice, however, it is more natural to first consider (1) which principal identification assumption to use (regardless of whether there is missing data) before considering (2) what to assume about the missingness. Conveniently, these two decisions can be made in this order because they can be made separately.

% To follow this order, we now pretend that the response probabilities $\varpi_{0c}(X)$ (and thus the mixture weights $\pi_c^*(X)$) are identified (or anticipate that they will be identified) under some specific missingness assumption, to focus on the key effect identification approach, i.e., the outcome assumption.

Our results point to a simple template for practice. When the principal identification assumption adopted is PI, no missingness assumption beyond LMAR is required, and the PCEs and ATE are identified by
\begin{alignat}{2}
    &\Delta_c
    &&=\frac{\E\big\{\pi_c(X)[\mu_{1c}(X)-\kappa_0^R(X)]\big\}}{\E[\pi_c(X)]},\label{eq:result.PI-PCEs}
    \\
    &\Delta
    &&=\E[\pi_1(X)\mu_{11}(X)+\pi_0(X)\mu_{10}(X)-\kappa_0^R(X)].\label{eq:result.PI-ATE}
\end{alignat}

When another principal identification assumption (e.g., ER or PIsens) is adopted, on the other hand, this only obtains \textit{contingent} identification of ${\color{red}\mu_{0c}(X)}$, and thus \textit{contingent} identification of the PCEs and ATE:
\begin{alignat}{2}
    &\Delta_c
    &&=\frac{\E\{\pi_c(X)[\mu_{1c}(X)-{\color{red}\mu_{0c}(X)}]\}}{\E[\pi_c(X)]},\label{eq:result.general-PCEs}
    \\
    &\Delta
    &&=\E\big\{\pi_1(X)[\mu_{11}(X)-{\color{red}\mu_{01}(X)}]+\pi_0(X)[\mu_{10}(X)-{\color{red}\mu_{00}(X)}]\big\},\label{eq:result.general-ATE}
\end{alignat}
where the formula for ${\color{red}\mu_{0c}(X)}$ (in the row corresponding to the principal identification assumption in Table~\ref{tab:identification}) involves the yet-to-be-identified mixture weights ${\color{violet}\pi_{0c}^R(X)}$. When a specific missingness mechanism is assumed, these mixture weights are identified (see Table~\ref{tab:missingprobs}), completing effect identification.

This template accommodates assumptions not covered here. The requirements of such assumptions are: (i) the principal identification assumption, when combined with A0, would identify the PCEs in the full-data setting; and (ii) the  specific missingness assumption, when combined with (A0, B0), identifies the response probabilities $\color{purple}{\varpi_{01}(X)}$ and $\color{purple}{\varpi_{00}(X)}$. 
% It is also advisable to choose an assumption that does not contradict the observed data distribution, a problem the stable response assumptions are subject to.
This combination works because (i) implies contingent (on the mixture weights) identification of $\color{red}{\mu_{0c}(X)}$ in the LMAR outcome setting.  

We reiterate that these two assumptions should be considered each on its own. (ER does not justify rER; and PI does not justify rPI.)
Of the two, the principal identification assumption is arguably much more important and needs to be carefully considered and well justified (and when in doubt sensitivity analysis is called for). The choice of a specific missingness assumption may not always be obvious, however, in which case different assumptions may be used and the range of results (rather than a single one) be considered.

\section{Illustration}\label{sec:illustration}

We return to Experience Corps and analyze the data to illustrate the various identification results under different assumptions above. As the focus of this paper is \textit{identification} and not estimation or inference, to avoid complicating the illustration, we simply use plug-in estimators based on all identification results. Other estimation strategies (e.g., weighting, multiply-robust, and in some cases some variation of two-stage least squares) could be pursued, but these require dedicated investigation for each identification formula, which is outside the scope of this paper. Also, without investigating variance estimation, we simply use the bootstrap to produce confidence intervals to give a sense of estimation uncertainty.

The sample used in this analysis includes 284 individuals in the intervention arm (352 randomized to intervention minus 68 who for different reasons withdrew or were excluded by the program after being randomized) and 339 in the control arm (350 randomized to control minus 11 who crossed over to the treatment arm). For the sake of illustration, we implement analyses based on three principal identification assumptions including ER, PI and PIsens-SMDe (the latter with sensitivity SMD ranging from $-$0.5 to 0.5); and consider four specific missingness assumptions including near-SNR, near-SCR, rPI and rPO.

Some comments can be made about these assumptions in the context of this example. Starting with principal identification assumptions, the original analysis of these data \citep{gruenewald2016BaltimoreExperienceCorps} used the ER assumption, but expressed concern that this assumption may not hold, because people below any positive cut point on volunteering hours were exposed to the intervention to some degree. This concern stands regardless of what covariates that may be considered in the analysis. The plausibility of the PI assumption, on the other hand, is specific to the baseline covariates. PI would be believable if the baseline covariates include important predictors of compliance in the intervention arm and important predictors of the outcome under control. That would make it more plausible that compliance type and outcome under control are rendered independent, although this assumption is untestable. In the illustrative analysis we use the same covariate set from the original analysis without modification (see Table~\ref{tab:EC-description}). This set includes age, sex, race, education, income, depressive symptoms (the Geriatric Depression Scale, \citeauthor{yesavage1982DevelopmentValidationGeriatric}, \citeyear{yesavage1982DevelopmentValidationGeriatric}), number of major morbidities, baseline measure of the outcome, and study year cohort (four cohorts); these variables are indeed predictive of both volunteering hours in the intervention arm and of the outcome under control. We note though that for a substantive (rather than illustrative) analysis, the choice of covariates for the PI assumption should be given thoughtful consideration based on theory and exploration of available baseline data. In case PI does not hold, the sensitivity analysis gives a sense of departure from PI-based effect estimates if conditional on covariates the outcome mean under control differs between compliers and noncompliers by at most an SMD of 0.5 and at least by an SMD of $-0.5$.

Regarding missingness, an important consideration for any application is whether the general LMAR assumption is plausible. LMAR says that missingness may depend on treatment assigned, compliance type and baseline covariates, but not on the outcome. LMAR should not be assumed if there is strong concern that the missingness may be caused by the outcome itself. In the present example, if one suspects that higher generativity (i.e., better outcome) \textit{makes} people more likely to respond to outcome assessment, then that violates LMAR. If on the other hand one thinks that there may be an association between outcome and missingness but believes that it is likely explained by baseline covariates (including baseline generativity, depressive symptoms, physical health and socio-economic status) and compliance type, then that is an argument for LMAR. The present paper only covers the LMAR case, so all the analyses here are based on LMAR.
Regarding the four specific missingness assumptions (near-SNR, near-SCR, rPI and rPO), in this case we do not have any substantive information to support one over another. We thus wish to see results under all these assumptions.

There is minimal missingness in two covariates, income (1.6\%) and morbidities (2.9\%). For simplicity and to keep the focus on the central topic, we fill in these missing values with a single stochastic imputation conditional on other analysis variables.

\begin{table}[t]
\caption{Baseline covariates in Experience Corps sample}\label{tab:EC-description}
\resizebox{\linewidth}{!}{%
\begin{tabular}{lrlrlrlrlrl}
    & \multicolumn{2}{c}{Full} & \multicolumn{2}{c}{Control} & \multicolumn{6}{c}{Intervention group} 
    \\\cline{6-11}
    & \multicolumn{2}{c}{sample} & \multicolumn{2}{c}{group} & \multicolumn{2}{c}{All} & \multicolumn{2}{@{}c@{}}{High participation} & \multicolumn{2}{c}{Low participation}
    \\ 
    & \multicolumn{2}{c}{n=623} & \multicolumn{2}{c}{n=339} & \multicolumn{2}{c}{n=284} & \multicolumn{2}{c}{n=168} & \multicolumn{2}{c}{n=116} \\ 
    \hline
    Age (mean (SD)) & 67.37 & (5.83) & 67.33 & (5.83) & 67.41 & (5.84) & 67.32 & (5.71) & 67.54 & (6.04) \\ 
    Sex: male (number (\%)) & 88 & (14.1) & 53 & (15.6) & 35 & (12.3) & 20 & (11.9) & 15 & (12.9) \\ 
    Race: Black (number (\%)) & 577 & (92.6) & 313 & (92.3) & 264 & (93.0) & 156 & (92.9) & 108 & (93.1) \\ 
    Education: some college (number (\%)) & 349 & (56.0) & 195 & (57.5) & 154 & (54.2) & 88 & (52.4) & 66 & (56.9) \\ 
    Income (number (\%)) &  &  &  &  &  \\ 
    ~~~$<$15K & 178 & (28.6) & 99 & (29.2) & 79 & (27.8) & 31 & (26.7) & 48 & (28.6) \\
    ~~~15 to $<$35K & 230 & (36.9) & 126 & (37.2) & 104 & (36.6) & 43 & (37.1) & 61 & (36.3) \\ 
    ~~~35K+ & 215 & (34.5) & 114 & (33.6) & 101 & (35.6) & 42 &(36.2) & 59 & (35.1) \\ 
    Major morbidities (mean (SD)) & 1.96 & (1.12) & 1.96 & (1.14) & 1.97 & (1.11) & 1.84 & (1.09) & 2.05 & (1.11) \\ 
    Depress symptom score (mean (SD)) & 1.12 & (1.67) & 1.17 & (1.83) & 1.07 & (1.46) & 0.96 & (1.37) & 1.22 & (1.57) \\ 
    Generative achievement (mean (SD)) & 5.17 & (0.81) & 5.16 & (0.83) & 5.18 & (0.80) & 5.15 & (0.80) & 5.22 & (0.79) \\ 
    \hline
\end{tabular}%
}
\bigskip
\bigskip
    \caption{Estimates of the average effects of Experience Corps on \textit{perceived generativity achievement} at 24 months for the high and low participation groups (CACE and NACE) and overall (ATE) under different pairings of principal identification assumption and specific missingness mechanism. Effects are on the mean difference scale (see $\Delta_c$ and $\Delta$ definitions in section \ref{subsec:notation}). 95\% bootstrap percentile confidence intervals are shown in brackets.}
    \label{tab:EC-application}
    \resizebox{\linewidth}{!}{%
    \begin{tabular}{lllll}
        \textit{Specific} & \multicolumn{4}{c}{\textit{Principal identification assumption}} 
        \\\cline{2-5}
        \textit{missingness} & Exclusion & Principal & \multicolumn{2}{l}{Sensitivity assumption (PIsens-SMDe)}
        \\\cline{4-5}
        \textit{mechanism} & restriction (ER) & ignorability (PI) & min SMD ($-0.5$)  & max SMD (0.5)
        \\\hline
        near-SNR & 
        \begin{tabular}{@{}l@{}l@{}}
            \textsc{cace}:~ & 0.19 (-0.10, 0.74)
            \\
            \textsc{nace}: & 0
            \\
            \textsc{ate}: & 0.11 (-0.06, 0.44)
        \end{tabular}
        &
        \multirow{10}{*}{
        \begin{tabular}{@{}l@{}l@{}}
            \textsc{cace}:~ &0.15 (0.05, 0.25)
            \\
            \textsc{nace}: &0.07 (-0.07, 0.20)
            \\
            \textsc{ate}: &0.12 (0.02, 0.22)
            \\
            \multicolumn{2}{@{}l@{}}{(not requiring a specific}
            \\
            \multicolumn{2}{@{}l@{}}{missingness assumption)}
        \end{tabular}}
        & 
        \begin{tabular}{@{}l@{}l@{}}
            \textsc{cace}:~ &\phantom{-}0.27 (0.15, 0.38)
            \\
            \textsc{nace}: &-0.12 (-0.24, 0.03)
            \\
            \textsc{ate}: &\phantom{-}0.11 (0.02, 0.22)
        \end{tabular}
        & 
        \begin{tabular}{@{}l@{}l@{}}
            \textsc{cace}:~ &0.02 (-0.06, 0.13)
            \\
            \textsc{nace}: &0.26 (0.08, 0.39)
            \\
            \textsc{ate}: &0.12 (0.02, 0.22)
        \end{tabular}
        \\\cline{1-2}\cline{4-5}
        near-SCR & 
        \begin{tabular}{@{}l@{}l@{}}
            \textsc{cace}:~ &0.18 (0.04, 0.33)
            \\
            \textsc{nace}: &0
            \\
            \textsc{ate}: &0.11 (0.02, 0.20)
        \end{tabular}
        &&
        \begin{tabular}{@{}l@{}l@{}}
            \textsc{cace}:~ &\phantom{-}0.24 (0.12, 0.35)
            \\
            \textsc{nace}: &-0.15 (-0.27, -0.01)
            \\
            \textsc{ate}: &\phantom{-}0.08 (-0.01, 0.18)
        \end{tabular}
        & 
        \begin{tabular}{@{}l@{}l@{}}
            \textsc{cace}:~ &0.06 (-0.03, 0.17)
            \\
            \textsc{nace}: &0.30 (0.12, 0.42)
            \\
            \textsc{ate}: & 0.15 (0.05, 0.26)
        \end{tabular}
        \\\cline{1-2}\cline{4-5}
        rPI & 
        \begin{tabular}{@{}l@{}l@{}}
            \textsc{cace}:~ &0.20 (0.04, 0.36)
            \\
            \textsc{nace}: &0
            \\
            \textsc{ate}: & 0.12 (0.02, 0.22)
        \end{tabular}
        &&
        \begin{tabular}{@{}l@{}l@{}}
            \textsc{cace}:~ &\phantom{-}0.28 (0.15, 0.38)
            \\
            \textsc{nace}: &-0.11 (-0.23, 0.03)
            \\
            \textsc{ate}: &\phantom{-}0.12 (0.02, 0.22)
        \end{tabular}
        & 
        \begin{tabular}{@{}l@{}l@{}}
            \textsc{cace}:~ &0.02 (-0.06,0.13)
            \\
            \textsc{nace}: &0.26 (0.09,0.38)
            \\
            \textsc{ate}: &0.12 (0.02, 0.22)
        \end{tabular}
        \\\cline{1-2}\cline{4-5}
        rPO & 
        \begin{tabular}{@{}l@{}l@{}}
            \textsc{cace}:~ & 0.18 (0.03,0.35)
            \\
            \textsc{nace}: &0
            \\
            \textsc{ate}: &0.10 (0.02, 0.21)
        \end{tabular}
        &&
        \begin{tabular}{@{}l@{}l@{}}
            \textsc{cace}:~ &\phantom{-}0.25 (0.13,0.36)
            \\
            \textsc{nace}: &-0.14 (-0.26,0.00)
            \\
            \textsc{ate}: &\phantom{-}0.09 (0.00, 0.19)
        \end{tabular}
        & 
        \begin{tabular}{@{}l@{}l@{}}
            \textsc{cace}:~ &0.04 (-0.04,0.16)
            \\
            \textsc{nace}: &0.28 (0.11,0.41)
            \\
            \textsc{ate}: &0.14 (0.04,0.24)
        \end{tabular}
        \\\hline
    \end{tabular}%
    }
\end{table}

We fit a model for compliance type in the intervention arm (the $\pi_c(X)$ model); three models for observing the outcome fit separately to the high and low participation groups in the intervention arm and to the control group (the $\varpi_{11}(X)$, $\varpi_{10}(X)$ and $\lambda_0(X)$ models); three outcome models also fit separately to the complete cases in these same three groups (the $\mu_{11}(X)$, $\mu_{10}(X)$ and $\kappa_0^R(X)$ models). We use logistic regression for the compliance type and response models. As the outcome is bounded (on a 1-to-6 scale) and is highly skewed to the upper end, we model it using GLM with a generalized logit link defined as $g(\mu)=\log[(\mu-l)/(h-\mu)]$ where $l=1,h=6$ are the lower and upper outcome bounds. (This departs from the original analysis, which uses a linear outcome model.)
The PIsens-SMDe analysis requires additionally estimating $\varsigma^{2R}(X):=\var(Y\mid X,Z=0,R=1)$. Using the quasi-likelihood approach, we assume this conditional variance is proportional to $[\kappa_0^R(X)-l][h-\kappa_0^R(X)]$. This is equivalent to assuming that $Y$, after being shifted and scaled to the $[0,1]$ interval, follows a quasibinomial model given covariates.

We then estimate the CACE, NACE and ATE under all assumption scenarios using formulas (\ref{eq:result.PI-PCEs}) (or (\ref{eq:result.general-PCEs})) and (\ref{eq:result.PI-ATE}) (or (\ref{eq:result.general-ATE})), where the functions of $X$ in these formulas are replaced with their estimated version computed at each $X$ value, and the expectations are estimated by averaging the integrands over all units in the dataset.
Code for this analysis is available in the R-package \texttt{latentMAR} at \url{https://github.com/trangnguyen74/latentMAR}.

Results are shown in Table~\ref{tab:EC-application}.
Effect estimates vary a great deal across the columns of the table, which is not surprising, because they are based on different principal identification assumptions. In applications generally one principal identification assumption is selected based on substantive considerations and perhaps is accompanied by one (or more) sensitivity analysis that deviates from that assumption. We pay more attention here to the variation within each principal identification approach.

First, as effect estimates under different specific missingness mechanisms reflect different identification results for response probabilities (Table~\ref{tab:missingprobs}), there is variation across the specific missingness mechanisms within each column of the table.
In this example, under ER, the variation in effect estimates assuming different missingness mechanisms is small, except the uncertainty is a lot larger for near-SNR. Under PI there is only one set of estimates because identification does not require missingness assumptions beyond LMAR. For PIsens-SMDe, under each end of the sensitivity parameter range, there is some variation in effect estimates across missingness assumptions, and this is more noticeable for the ATE estimate (which is a weighted average of the CACE and NACE estimates). All these variations in point estimates, however, are small relative to their uncertainty. This, overall, means that for this specific study, effect estimates are not very sensitive to these specific missingness assumption.

Second, consider scenarios in Table~\ref{tab:EC-application} where the outcome turns out to be MAR: when PI is assumed, and when rPI is combined with any principal identification assumption. In all these cases, the ATE estimate is the same, even though we have estimated it in different ways based on different principal identification assumptions. This is not surprising because the ATE is identified under MAR, by (\ref{eq:result.PI-ATE}), and this is regardless of how we arrive at MAR. The CACE and NACE estimates, however, vary with the principal identification assumption.

\section{Discussion}\label{sec:conclusions}

This paper provides a unifying view of the nonidentifiability of PCEs (and ATE) with LMAR outcomes under treatment assignment ignorability: this involves two connected mixture equations with four unknowns, where two of the unknowns (conditional outcome means under control) need to be identified for the effects to be identified. This informs a template for effect identification: combining a principal identification assumption with a missingness assumption that is more specific than LMAR to solve the pair of mixture equations. Applying this template, effect identification formulas are derived (under different principal identification assumptions) for several specific missingness assumptions; these fall into the categories of assumptions that compare missingness across treatment conditions (type 1) and assumptions that compare missingness across principal strata (type 2). These results are immediately useful for cases where those assumptions are relevant, and they serve as examples for the utility of the template.

Some specific results should be mentioned. Under either PI (the principal identification assumption) or rPI (the specific missingness assumption) -- combined with LMAR -- the outcome missingness turns out to be MAR. For the first case (PI), a specific missingness assumption is not required for effect identification. However, if a PI-based analysis is paired with a sensitivity analysis for PI violation, the sensitivity analysis requires a specific missingness assumption.

An incidental finding emerged that deserves some attention. The two existing missingness assumptions SNR/rER and SCR can imply implausible response probabilities, which is a contradiction with the observed data distribution. As a simple fix, we propose that these assumptions be replaced with a modified version (near-SNR and near-SCR) that avoids this problem. More generally, of the two types of specific missingness assumptions, type 1 (where SNR and SCR belong) is particularly prone to this problem. This is because assumptions of this type fix one of the unknown response probabilities in the response mixture equation, which concerns \textit{the control condition}, to a value defined entirely based on \textit{the intervention condition}; and this value may be unrealistic given the actual degree of missingness under control (of compliers and noncompliers combined). For this problem, the same modification strategy is generally useful.
Type 2 assumptions are less prone to this problem because, unlike type 1, they do not link missingness rates under control to those under intervention.
The two assumptions rPI and rPO listed under type 2 are in fact free of this problem, and so are their linear combinations. 
Note though that this guarantee for rPO is due to the use of the odds ratio to connect complier and noncomplier response probabilities. Replacing the odds ratio with another parameter, however, can induce conflict if that parameter is agnostic of the binary variable's bounds. For example, if the probability ratio is used instead, and if for some $X$ values the observed response under control is high but the assumption-implied response probability ratio is large, then one of the stratum-specific response probabilities can exceed 1.
Therefore caution should be taken to avoid this kind of situation. In our opinion though the problem affecting type 1 assumptions is more serious.

% Some of the specific results indicate that, while a specific missingness assumption is generally required, PI is a special case that does not require this. Also, under either PI or rPI (combined with LMAR), the outcome missingness turns out to be MAR. 

Going back to the big picture, with the proposed template, this work expands options for missingness assumptions and new identification results. This calls for additional work on estimation and inference.
% Much more work is needed in this area; we list several directions here. First, this paper addressed identification but estimation and inference need to be investigated. 
As many methods have been developed for PCE estimation under different principal identification assumptions in the full-data setting, it is desirable that those methods are adapted to handle outcome missingness, and that the adaptation is flexible enough to allow different specific missingness assumptions when these are required. Inference under the near-SNR and near-SCR assumptions is likely challenging and should be investigated.

Beyond this, there are other directions for future extension.
A perhaps simple extension is to allow the conditioning covariate sets to differ between causal identification assumptions and missingness assumptions.
Another is to investigate how to handle the combination of LMAR outcome with missingness in covariates and in some cases in treatment assignment and compliance type; in these situations the identification challenge is more complex.
It is also useful to adapt the current template (and perhaps specific results) to the two-sided noncompliance setting, which is also common.
And last but not least, research is needed on how to handle the situation where one does not assume the outcome missingness is LMAR.

\bibliography{reference.bib}

%%%%%%%%%%%%%%%%%%%%%%%%%%%%%%%%%%%%%%%%

% \newpage
\clearpage
\pagenumbering{arabic}% resets `page` counter to 1
\renewcommand*{\thepage}{A\arabic{page}}
\appendix
% \noindent\textbf{\Large Appendix}

% \section*{Supplementary Material}

% for

% Nguyen, Carlson and Stuart. (2023). Identification of complier and noncomplier average causal effects in the presence of \textit{latent} missing-at-random (LMAR) outcomes: a unifying view and choices of assumptions

% \bigskip

\section*{Appendix}

\begin{lemma}\label{lm:1}
If $A\independent(B,C)$ then $A\independent B\mid C$.
\end{lemma}
\bigskip

%%%%%%%%%%%%%%%%%%%%%%

\begin{proof}[Proof of Lemma \ref{lm:1}]
\hfill
\smallskip

\noindent
    First, note that if $A\independent(B,C)$ then $A\independent C$. This is because
    \begin{align*}
        \P(A\mid C)
        &=\E[\P(A\mid B,C)\mid C]
        && \text{(law of total probability)}
        \\
        &=\E[\P(A)\mid C]
        && (A\independent(B,C))
        \\
        &=\P(A).
    \end{align*}
    Then $A\independent B\mid C$ follows, because
    \begin{align*}
        \P(B\mid C)
        &=\frac{\P(B,C)}{\P(C)}
        && \text{(Bayes' rule}
        \\
        &=\frac{\P(B,C\mid A)}{\P(C\mid A)}
        && (A\independent(B,C)~\text{and}~A\independent C)
        \\
        &=\P(B\mid A,C).
        && \text{(Bayes' rule)}
    \end{align*}
\end{proof}
\bigskip

%%%%%%%%%%%%%%%%%%%%%%%%%%%%%%%

\begin{proof}[Proof of (\ref{eq:Delta.c})]
\hfill
\smallskip

\noindent
    We repeat what needs to be shown,
    \begin{align}
        \overbrace{\E[Y_1-Y_0\mid C=c]}^{\textstyle=:\Delta_c}=\frac{\E\{\overbrace{\E[Y_1-Y_0\mid X,C=c]}^{\textstyle=:\delta_c(X)}\P(C=c\mid X)\}}{\E[\P(C=c\mid X)]}.\tag{\ref{eq:Delta.c}}
    \end{align}
    For $z=1,0$,
% \end{proof}

% The proofs of (\ref{eq:delta.c}) and (\ref{eq:start}) (including Lemma \ref{lm:1}) appear in in the Web Appendix of \cite{nguyen2023SensitivityAnalysisPrincipal}, but are repeated here for convenience.

% \begin{proof}[Proof of (\ref{eq:tau.z}) and (\ref{eq:start})]\hfill

% We repeat what needs to be shown:
% \begin{align*}
%     \E[Y_z\mid C=c]
%     &=\frac{\E\big\{\E[Y_z\mid X,C=c]\P(C=c\mid X)\big\}}{\E[\P(C=c\mid X)]},\tag{\ref{eq:tau.z}}
%     \\
%     \E[Y_z\mid C=c]
%     &=\frac{\E\big\{\E[Y\mid X,C=c,Z=z]\P(C=c\mid X,Z=1)\big\}}{\E[\P(C=c\mid X,Z=1)]}.\tag{\ref{eq:start}}
% \end{align*}
% We first show (\ref{eq:tau.z}):
\begin{align}
    \E[Y_z\mid C=c]
    &=\E\{\E[Y_z\mid X,C=c]\mid C=c\} 
    && \text{(iterated expectation)}\nonumber
    \\
    % &=\int\E[Y_z\mid X=x,C=c]f_{X\mid C}(x\mid c)dx
    % && \text{(writing expectation as integration)}
    % \\
    % &=\int\E[Y_z\mid X=x,C=c]\frac{f_{X\mid C}(x\mid c)}{f_X(x)}f_X(x)dx
    % \\
    &=\E\left\{\E[Y_z\mid X,C=c]\frac{\P(X)}{\P(X)}\mid C=c\right\}\nonumber
    \\
    &=\E\left\{\E[Y_z\mid X,C=c]\frac{\P(X\mid C=c)}{\P(X)}\right\}\nonumber
    % && \text{(writing integration as expectation)}
    \\
    &=\E\left\{\E[Y_z\mid X,C=c]\frac{\P(X, C=c)/\P(C=c)}{\P(X,C=c)/\P(C=c\mid X)}\right\}\nonumber
    && \text{(Bayes' rule)}
    \\
    &=\E\left\{\E[Y_z\mid X,C=c]\frac{\P(C=c\mid X)}{\P(C=c)}\right\}\nonumber
    \\
    &=\frac{\E\big\{\E[Y_z\mid X,C=c]\P(C=c\mid X)\big\}}{\P(C=c)}\nonumber
    \\
    &=\frac{\E\big\{\E[Y_z\mid X,C=c]\P(C=c\mid X)\big\}}{\E[\P(C=c\mid X)]}.\label{eq:tau.z}
\end{align}
Plugging in $z=1$ and $z=0$ to (\ref{eq:tau.z}) and taking the difference obtains (\ref{eq:Delta.c}).
\end{proof}
\bigskip

%%%%%%%%%%%%%%%%%%%%%%%%%%%%%

\begin{proof}[Proof of (\ref{eq:delta.c})]
\hfill
\smallskip

\noindent To show (\ref{eq:delta.c}), we need to show that for $z=1,0$,
    \begin{align}
        \E[Y_z\mid X,C=c]=\E[Y\mid X,C=c,Z=z].
    \end{align}
    This is true because
\begin{align*}
    \E[Y_z\mid X,C=c]
    =\E[Y_z\mid X,C=c,Z=z]
    =\overbrace{\E[Y\mid X,C=c,Z=z]}^{\textstyle=:\mu_{cz}(X)},\label{eq:mu.zc}
\end{align*}
where the first equality is due to $Z\independent Y_z\mid X,C$ (which is obtained implied by by Lemma \ref{lm:1} from the unconfoundedness part of A0), and the second equality is due to consistency (invoked when defining the effects) and positivity (also part of A0).  
\end{proof}
\bigskip

%%%%%%%%%%%%%%%%%%%%%%%%%%%%%%%

\begin{proof}[Proof of (\ref{eq:mix.wts})] 
\hfill
\smallskip

\noindent    We repeat what needs to be shown:
    \begin{align}
        \overbrace{\P(C=c\mid X,Z=0,R=1)}^{\textstyle=:\pi_{0c}^R(X)}
        =\overbrace{\P(C=c\mid X,Z=1)}^{\textstyle=:\pi_c(X)}\frac{\overbrace{\P(R=1\mid X,Z=0,C=c)}^{\textstyle=:\varpi_{0c}(X)}}{\underbrace{\P(R=1\mid X,Z=0)}_{\textstyle=:\lambda_0(X)}}.
        \tag{\ref{eq:mix.wts}}
    \end{align}
    We start with the LHS:
    \begin{align*}
        \P(C=c\mid X,Z=0,R=1)
        &=\frac{\P(C=c,R=1\mid X,Z=0)}{\P(R=1\mid X,Z=0)}
        && \text{(Bayes' rule)}
        \\
        &=\frac{\P(C=c\mid X,Z=0)\P(R=1\mid X,Z=0,C=c)}{\P(R=1\mid X,Z=0)}
        \\
        &=\P(C=c\mid X,Z=1)\frac{\P(R=1\mid X,Z=0,C=c)}{\P(R=1\mid X,Z=0)}
        && (Z\independent C\mid X)
        \\
        &=\text{RHS}.
    \end{align*}
\end{proof}
\bigskip

%%%%%%%%%%%%%%%%%%%%%%%%%%%%%%%

\begin{proof}[Elaboration of a point in section \ref{sec:specific-miss-type1}]
\hfill
\smallskip

\noindent We copy here the second paragraph under \textbf{A limitation of stable response} in section \ref{sec:specific-miss-type1}: 

``While there now exists a literature following and extending on the original rER assumption \cite[e.g.,][and work cited in Section \ref{sec:intro}]{dunn2005EstimatingTreatmentEffects,zhou2006ITTAnalysisRandomized,taylor2009MultipleImputationMethods,chen2009IdentifiabilityEstimationCausal,lui2010NotesOddsRatio,mealli2012RefreshingAccountPrincipal,chen2015SemiparametricInferenceComplier}, surprisingly, this risk of implausible response probabilities seems to be under-appreciated. This might be because previous work, when deriving $\color{red}\mu_{0c}(X)$, did this only for one setting (compound ER), and thus did not need to derive (let alone scrutinize) the response probabilities $\color{purple}\varpi_{0c}(X)$ as an intermediate step (see more details on this in the Appendix).''

We now make this concrete. As mentioned in the paper, SNR/rER implies ${\color{purple}\varpi_{00}(X)}=\varpi_{10}(X)$ and ${\color{purple}\varpi_{01}(X)}=\frac{\lambda_0(X)-\pi_0(X)\varpi_{10}(X)}{\pi_1(X)}$. This combined with our result for ${\color{red}\mu_{01}(X)}$ under ER 
in Table~\ref{tab:identification} 
imply (after some algebra) that
\begin{align*}
    {\color{red}\mu_{01}(X)}
    % &=\mu_{10}(X)+\frac{\kappa_0^R(X)-\mu_{10}(X)}{1-\pi_0(X)\frac{\varpi_{10}(X)}{\lambda_0(X)}}
    % \\
    % &=\mu_{10}+\frac{\lambda_0[\kappa_0^R-\mu_{10}]}{\lambda_0-\pi_0\varpi_{10}}
    % \\
    &=\frac{\lambda_0(X)\kappa_0^R(X)-\pi_0(X)\varpi_{10}(X)\mu_{10}(X)}{\lambda_0(X)-\pi_0(X)\varpi_{10}(X)},\label{eq:compoundER}
\end{align*}
which coincides with the result in Lemma 1 in \cite{frangakis1999AddressingComplicationsIntentiontotreat}, except we condition on covariates. Their proof, translated to our current language, involves a step that replaces $\pi_1(X){\color{purple}\varpi_{01}(X)}$ with $\lambda_0(X)-\pi_0(X)\varpi_{10}(X)$ (to obtain the denominator above). This is just shy of deriving ${\color{purple}\varpi_{01}(X)}$, but deriving ${\color{purple}\varpi_{01}(X)}$ was not the goal in their proof. Therefore, the fact that this can imply a value for this probability that is out of range was not noticed.

\end{proof}
\bigskip

%%%%%%%%%%%%%%%%%%%%%%%%%%%%%%%%%

\begin{proof}[Proof of results in Table \ref{tab:missingprobs}]
\hfill
\smallskip

\noindent Table \ref{tab:missingprobs} shows results for all four quantities ($\varpi_{01}(X)$, $\varpi_{00}(X)$, $\pi_{01}^R(X)$ and $\pi_{00}^R(X)$) for completeness. The derivation is quite tedious, so here we just focus on obtaining the mixture weights, and since the two mixture weights sum to 1, we just need to obtain one of them.

    \vspace{.5em}\noindent
    \underline{Part 1}: If we were to assume SNR/rER ($R\independent Z\mid X,C=0$), we would equate ${\color{purple}\varpi_{00}(X)}:=\P(R=1\mid X,Z=0,C=0)$ to $\varpi_{10}(X):=\P(R=1\mid X,Z=1,C=0)$. This, combined with (\ref{eq:response.mixture}), would imply ${\color{purple}\varpi_{01}(X)}$ is equal to $\frac{\lambda_0(X)-\pi_0(X)\varpi_{10}(X)}{\pi_1(X)}$. Switching to \textit{near}-SNR, we have the identification result for ${\color{purple}\varpi_{01}(X)}$ in Table \ref{tab:missingprobs}.
    With that result for ${\color{purple}\varpi_{01}(X)}$, based on (\ref{eq:mix.wts}) with $c=0$, we identify the mixture weights,
    \begin{align*}
        {\color{violet}\pi_{01}^R(X)}=\pi_1(X)\frac{{\color{purple}\varpi_{01}(X)}}{\lambda_0(X)}.
    \end{align*}

    \vspace{.5em}\noindent
    \underline{Part 2}: Near-SCR is a miror image of near-SNR, so the proof is trivially similar.

    \vspace{.5em}\noindent
    \underline{Part 3}: Under rPI (formally $R\independent C\mid X,Z=0$), for $c=1,0$,
    $$\overbrace{\P(R=1\mid X,Z=0,C=c)}^{\textstyle=:{\color{purple}\varpi_{0c}(X)}}=\overbrace{\P(R=1\mid Z=0)}^{\textstyle=:\lambda_0(X)}.$$
    Combining this with (\ref{eq:mix.wts}) obtains
    $${\color{violet}\pi_{0c}^R(X)}=\pi_c(X).$$

    \vspace{.5em}\noindent
    \underline{Part 4}: The proof for rPO is most complicated. (It is very similar to GOR part of the proof of Proposition 5 in \cite{nguyen2023SensitivityAnalysisPrincipal}.) rPO states that
    \begin{align}
        \frac{{\color{purple}\varpi_{01}(X)}/[1-{\color{purple}\varpi_{01}(X)}]}{{\color{purple}\varpi_{00}(X)}/[1-{\color{purple}\varpi_{00}(X)}]}
        =\frac{\varpi_{11}(X)/[1-\varpi_{11}(X)]}{\varpi_{10}(X)/[1-\varpi_{10}(X)]}=:\varrho(X).
    \end{align}
    Combining this with the response mixture equation
    \begin{align}
        \pi_1(X){\color{purple}\varpi_{01}(X)}+\pi_0(X){\color{purple}\varpi_{01}(X)}=\lambda_0(X),\tag{\ref{eq:response.mixture}}
    \end{align}
    we have two equations with two unknowns. As all quantities in these two equations condition on $X$ only, we suppress the $(X)$ notation. Also, let the two unknowns ${\color{purple}\varpi_{01}(X)}$ and ${\color{purple}\varpi_{00}(X)}$ be represented by $u$ and $v$, respectively. Our two equations are
    \begin{align*}
        \begin{cases}
            \frac{u/(1-u)}{v/(1-v)}=\varrho
            \\
            \pi_1u+(1-\pi_1)v=\lambda_0
        \end{cases}
    \end{align*}
    subject to the condition $u,v\in[0,1]$.

    For ($X$ values such that) $\varrho=1$, we have $u=v=\lambda_0$. Now we focus on ($X$ values such that) $\varrho\neq 1$.

    The second equation gives $v=\mfrac{\lambda_0-\pi_1u}{1-\pi_1}$. Plugging this into the first equation, we obtain (after some algebra)
    $$\pi_1(\varrho-1)u^2-[(\pi_1+\lambda_0)(\varrho-1)+1]u+\lambda_0\varrho=0.$$
    This quadratic equation has two roots
    \begin{align*}
        u_1=\frac{[(\pi_1+\lambda_0)(\varrho-1)+1]+\sqrt{d}}{2\pi_1(\varrho-1)},~~~
        u_2=\frac{[(\pi_1+\lambda_0)(\varrho-1)+1]-\sqrt{d}}{2\pi_1(\varrho-1)},
    \end{align*}
    where
    $$d=[(\pi_1+\lambda_0)(\varrho-1)+1]^2-4\pi_1\lambda_0\varrho(\varrho-1).$$
    These two roots for $u$ respectively imply two values for $v$:
    \begin{align*}
        v_1=\frac{[(\lambda_0-\pi_1)(\varrho-1)-1]-\sqrt{d}}{2(1-\pi_1)(\varrho-1)},~~~v_2=\frac{[(\lambda_0-\pi_1)(\varrho-1)-1]+\sqrt{d}}{2(1-\pi_1)(\varrho-1)}.
    \end{align*}
    and we note (after some algebra) another helpful expression of $d$
    $$d=[(\lambda_0-\pi_1)(\varrho-1)-1]^2+4(1-\pi_1)\lambda_0(\varrho-1).$$

    Now we check these candidates for $u$ and $v$ against the condition $u,v\in(0,1)$. If $\varrho>1$, it can be shown that $d\geq[(\lambda_0-\pi_1)(\varrho-1)-1]^2$, which implies $v_2\geq0$ but $v_1<0$, ruling out the candidate $v_1$. If $\rho<1$ then $d\geq[(\pi_1+\lambda_0)(\varrho-1)+1]^2$, which implies $u_2\geq0$ but $u_1<0$, ruling out the candidate $u_1$. In both cases, the choice left is
    \begin{align*}
        u=u_2=\frac{[(\pi_1+\lambda_0)(\varrho-1)+1]-\sqrt{d}}{2\pi_1(\varrho-1)},
        ~~~
        v_2=\frac{[(\lambda_0-\pi_1)(\varrho-1)-1]+\sqrt{d}}{2(1-\pi_1)(\varrho-1)}.
    \end{align*}
    That $u_2$ and $v_2$ are also $\leq1$ is clear from the first equation, which says that their weighted average is $\leq1$. The solution to the set of two equations is thus $(u_2,v_2)$.

    Let $\varrho_1(X):=\varrho(X)$, $\varrho_0(X):=1/\varrho(X)$.

    The $u_2$ formula is spelled out as
    \begin{align*}
        {\color{purple}\varpi_{01}(X)}=\frac{\overbrace{[\pi_1(X)+\lambda_0(X)][\varrho_1(X)-1]+1}^{\textstyle=:\gamma_1(X)}-\overbrace{\sqrt{[\gamma_1(X)]^2-4\pi_1(X)\lambda_0(X)\varrho_1(X)[\varrho_1(X)-1]}}^{\textstyle=:\omega_1(X)}}{2\pi_1(X)[\varrho_1(X)-1]}.
    \end{align*}
    Leveraging symmetry, we have the result
    \begin{align*}
        {\color{purple}\varpi_{0c}(X)}=\frac{\overbrace{[\pi_c(X)+\lambda_0(X)][\varrho_c(X)-1]+1}^{\textstyle=:\gamma_c(X)}-\overbrace{\sqrt{[\gamma_c(X)]^2-4\pi_c(X)\lambda_0(X)\varrho_c(X)[\varrho_c(X)-1]}}^{\textstyle=:\omega_c(X)}}{2\pi_c(X)[\varrho_c(X)-1]}.
    \end{align*}
    Combining this with (\ref{eq:mix.wts}) obtains
    \begin{align*}
        {\color{violet}\pi_{0c}^R(X)}=\frac{\gamma_c(X)-\omega_c(X)}{2[\varrho_c(X)-1]\lambda_0(X)}.
    \end{align*}
\end{proof}
\bigskip

%%%%%%%%%%%%%%%%%%%%%%%%%%%%%%%%%

\begin{proof}[Proof of MAR under rPI]
\hfill

    \begin{align*}
        \P(R=1\mid &X,Z=0,Y)
        \\
        &=\sum_c\P(R=1\mid X,Z=0,C=c,Y)\P(C=c\mid X,Z=0,Y)
        && \text{(total probability)}
        \\
        &=\sum_c\P(R=1\mid X,Z=0,C=c)\P(C=c\mid X,Z=0,Y)
        && \text{(LMAR)}
        \\
        &=\sum_c\P(R=1\mid X,Z=0)\P(C=c\mid X,Z=0,Y)
        && \text{(rPI)}
        \\
        &=\P(R=1\mid X,Z=0)\sum_c\P(C=c\mid X,Z=0,Y)
        \\
        &=\P(R=1\mid X,Z=0).
    \end{align*}
\end{proof}
\bigskip

%%%%%%%%%%%%%%%%%%%%%%%%%%%%%%%%%

\begin{proof}[Proof of results in Table \ref{tab:identification}]\hfill

    \vspace{.5em}\noindent
    \underline{Part 1}: Under ER (formally $Z\independent Y\mid X,C=0$),
    $$\overbrace{\E[Y\mid X,Z=0,C=0]}^{\textstyle=:{\color{red}\mu_{00}(X)}}=\overbrace{\E[Y\mid X,Z=1,C=0]}^{\textstyle=:\mu_{10}(X)}.$$
    Combining this with the outcome mixture equation
    \begin{align}
        \pi_{01}^R(X){\color{red}\mu_{01}(X)}+[1-\pi_{01}^R(X)]{\color{red}\mu_{00}(X)}=\kappa_0^R(X)\tag{\ref{eq:outcome.mixture}}
    \end{align}
    obtains
    $${\color{red}\mu_{01}(X)}=\mu_{10}(X)+\frac{\kappa_0^R(X)-\mu_{10}(X)}{\pi_{01}^R(X)}.$$

    \vspace{.5em}\noindent
    \underline{Part 2}: The PI case has been clearly argued in the text of the paper.

    \vspace{.5em}\noindent
    \underline{Part 3}: The results for the three PIsens assumptions are a simple translation from the results for the full data case \citep[][Propositions 5 and 8]{nguyen2023SensitivityAnalysisPrincipal} to the current case with missing outcome. Those prior identification results solve the full-data outcome mixture equation
    \begin{align*}
        \pi_1(X){\color{red}\mu_{01}(X)}+\pi_0(X)]{\color{red}\mu_{00}(X)}=\kappa_0(X)\tag{\ref{eq:full.mixture}}
    \end{align*}
    under each of the PIsens assumptions.
    As the missing outcome version of this equation is 
    \begin{align}
        \pi_{01}^R(X){\color{red}\mu_{01}(X)}+\pi_{00}^R(X){\color{red}\mu_{00}(X)}=\kappa_0^R(X),\tag{\ref{eq:outcome.mixture}}
    \end{align}
    the translation of the identification results is simply to replace $\kappa_0(X)$ with $\kappa_0^R(X)$ and $\pi_c(X)$ with $\pi_{0c}^R(X)$.
\end{proof}
\bigskip

%%%%%%%%%%%%%%%%%%%%%%%%%%%%%%%%%

\begin{proof}[Proof of MAR under PI]
\hfill

    \begin{align*}
        \P(R=1\mid &X,Z=0,Y)
        \\
        &=\sum_c\P(R=1\mid X,Z=0,C=c,Y)\P(C=c\mid X,Z=0,Y)
        && \text{(total probability)}
        \\
        &=\sum_c\P(R=1\mid X,Z=0,C=c)\P(C=c\mid X,Z=0,Y)
        && \text{(LMAR)}
        \\
        &=\sum_c\P(R=1\mid X,Z=0,C=c)\P(C=c\mid X,Z=0)
        && \text{(PI)}
        \\
        &=\P(R=1\mid X,Z=0).
    \end{align*}
\end{proof}

\end{document}